\documentclass[final,12pt]{article}
\usepackage{hyperref}
\usepackage{amssymb,amsmath,amsthm}
\usepackage{tikz}
\usetikzlibrary{arrows,backgrounds,decorations.pathmorphing,decorations.pathreplacing,positioning,fit}
\usepackage[all]{xy}
\usepackage{enumerate}
\usepackage[conditional,light,first,bottomafter]{draftcopy}
\draftcopyName{DRAFT\space\today}{130}
\draftcopySetScale{65}
\usepackage[letterpaper,hmargin=3.7cm,vmargin=3.3cm]{geometry}
\usepackage{setspace}
\makeatletter
\renewcommand{\section}{\@startsection%
{section}%
{1}%
{0em}%
{1.7em}%
{1.2em}%
{\normalfont\large\centering\bfseries}}
\renewcommand{\@seccntformat}[1]%
{\csname the#1\endcsname.\hspace{0.5em}}
\makeatother
\renewcommand{\thesection}{\arabic{section}}

\numberwithin{equation}{section}
\renewcommand\appendix{\par
\setcounter{section}{0}%
\setcounter{subsection}{0}%
\setcounter{theorem}{0}
\setcounter{table}{0}
\setcounter{figure}{0}
\gdef\thetable{\Alph{table}}
\gdef\thefigure{\Alph{figure}}
\section*{Appendix}
\gdef\thesection{\Alph{section}}
\setcounter{section}{1}}

\newtheorem{theorem}{Theorem}[section]
\newtheorem{proposition}[theorem]{Proposition}
\newtheorem{lemma}[theorem]{Lemma}

\theoremstyle{definition}
\newtheorem{definition}[theorem]{Definition}
\newtheorem{remark}[theorem]{Remark}
\newtheorem{notation}{N}
\newtheorem{hypothesis}{H}

\newcommand{\abs}[1]{\left|#1\right|}
\newcommand{\norm}[1]{\left\|#1\right\|}
\newcommand{\inner}[2]{\left\langle#1,#2\right\rangle}
\newcommand{\cc}[1]{\overline{#1}}
\newcommand{\reals}{\mathbb{R}}
\newcommand{\nats}{\mathbb{N}}
\newcommand{\integers}{\mathbb{Z}}
\newcommand{\complex}{\mathbb{C}}
\newcommand{\eval}[1]{\upharpoonright_{#1}}
\newcommand{\convergesto}[1]{\xrightarrow[#1\to\infty]{}}
\newcommand{\convergestozero}[1]{\xrightarrow[#1\to 0]{}}
\newcommand{\convergestozeros}[2]{\xrightarrow[#1\to 0]{#2}}
\DeclareMathOperator{\im}{Im}
\DeclareMathOperator{\dom}{dom}
\DeclareMathOperator*{\res}{Res}
\DeclareMathOperator{\Span}{span}
\DeclareMathOperator{\card}{card}

\begin{document}
\begin{titlepage}
\title{Inverse problems for Jacobi operators IV:\\
Interior mass-spring perturbations of semi-infinite systems
\footnotetext{%
Mathematics Subject Classification(2010):
34K29,  % Inverse problems (34-XX Ordinary differential equations 34Kxx
       % Functional-differential and differential-difference
       % equations)
47A75, % Eigenvalue problems (47-XX Operator theory 47Axx General
       % theory of linear operators)
47B36, % Jacobi (tridiagonal) operators (matrices) and generalizations
       % (47-XX Operator theory 47Bxx Special classes of linear operators)
70F17, % Inverse problems (70-XX Mechanics of particles and systems
       % 70Fxx Dynamics of a system of particles, including celestial mechanics)
}
\footnotetext{%
Keywords:
Infinite mass-spring system;
Jacobi matrices;
Two-spectra inverse problem.
}
\hspace{-8mm}
\thanks{%
Research partially supported by UNAM-DGAPA-PAPIIT IN105414
}%
}
\author{
\textbf{Rafael del Rio}
\\
%% ----- Institution --------
\small Departamento de F\'{i}sica Matem\'{a}tica\\[-1.6mm]
\small Instituto de Investigaciones en Matem\'aticas Aplicadas y en Sistemas\\[-1.6mm]
\small Universidad Nacional Aut\'onoma de M\'exico\\[-1.6mm]
\small C.P. 04510, M\'exico D.F.\\[-1.6mm]
\small \texttt{delrio@iimas.unam.mx}
\\[2mm]
\textbf{Mikhail Kudryavtsev}
\\
%% ----- Institution ---------
\small Department of Mathematics\\[-1.6mm]
\small Institute for Low Temperature Physics and Engineering\\[-1.6mm]
\small Lenin Av. 47, 61103\\[-1.6mm]
\small Kharkov, Ukraine\\[-1.6mm]
\small\texttt{kudryavtsev@onet.com.ua}
\\[2mm]
\textbf{Luis O. Silva}
\\
%% ----- Institution --------
\small Departamento de F\'{i}sica Matem\'{a}tica\\[-1.6mm]
\small Instituto de Investigaciones en Matem\'aticas Aplicadas y en Sistemas\\[-1.6mm]
\small Universidad Nacional Aut\'onoma de M\'exico\\[-1.6mm]
\small C.P. 04510, M\'exico D.F.\\[-1.6mm]
\small \texttt{silva@iimas.unam.mx}
}
%%%%%%%%
\date{}
\maketitle
\vspace{-4mm}
\begin{center}
\begin{minipage}{5in}
  \centerline{{\bf Abstract}}
\bigskip
This work gives results on the interplay of the spectra of two Jacobi
operators corresponding to an infinite mass-spring system and a
modification of it obtained by changing one mass and one spring of the
system. It is shown that the system can be recovered from these two
spectra. Necessary and sufficient conditions for two sequences to be
the spectra of the mass-spring system and the perturbed one are
provided.
\end{minipage}
\end{center}
\thispagestyle{empty}
\end{titlepage}
%%%%%%%%%%%%%%%%%%%%%%%%%%%%%%
\section{Introduction}
\label{sec:intro}
Inverse spectral problems are concerned with the quest of information
determining an operator from its spectral data. These problems have
various applications in physics and other sciences. Usually, we do not
possess all the information that defines the operator modeling a
certain physical system, however it is possible to measure physical
quantities related to the spectrum of the operator and use these data
to gain some information about the operator, thence about the system.

The kind of inverse spectral problem studied in the present work is
the so called two spectra inverse problem in which one is given the
spectra of an operator and a perturbation of it with the goal of
recovering the operator from these two spectra.

\begin{figure}[h]
\begin{center}
\begin{tikzpicture}
  [mass1/.style={rectangle,draw=black!80,fill=black!13,thick,inner sep=0pt,
   minimum size=7mm},
   mass2/.style={rectangle,draw=black!80,fill=black!13,thick,inner sep=0pt,
   minimum size=5.7mm},
   mass3/.style={rectangle,draw=black!80,fill=black!13,thick,inner sep=0pt,
   minimum size=7.7mm},
   wall/.style={postaction={draw,decorate,decoration={border,angle=-45,
   amplitude=0.3cm,segment length=1.5mm}}}]
  \node (mass3) at (7.1,1) [mass3] {\footnotesize$m_3$};
  \node (mass2) at (4.25,1) [mass2] {\footnotesize$\,m_2$};
  \node (mass1) at (2.2,1) [mass1] {\footnotesize$m_1$};
\draw[decorate,decoration={coil,aspect=0.4,segment
  length=2.1mm,amplitude=1.8mm}] (0,1) -- node[below=4pt]
{\footnotesize$k_1$} (mass1);
\draw[decorate,decoration={coil,aspect=0.4,segment
  length=1.5mm,amplitude=1.8mm}] (mass1) -- node[below=4pt]
{\footnotesize$k_2$} (mass2);
\draw[decorate,decoration={coil,aspect=0.4,segment
  length=2.5mm,amplitude=1.8mm}] (mass2) -- node[below=4pt]
{\footnotesize$k_3$} (mass3);
\draw[decorate,decoration={coil,aspect=0.4,segment
  length=2.1mm,amplitude=1.8mm}] (mass3) -- node[below=4pt]
{\footnotesize$k_4$} (9.3,1);
\draw[line width=.8pt,loosely dotted] (9.4,1) -- (9.8,1);
\draw[line width=.5pt,wall](0,1.7)--(0,0.3);
\end{tikzpicture}
\end{center}
\caption{Semi-infinite mass-spring system}\label{fig:1}
\end{figure}
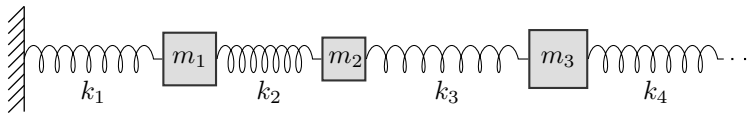

Consider the linear semi-infinite mass spring system illustrated in
Fig. 1. This mechanical system, with masses
$\{m_j\}_{j=1}^\infty$ and spring constants $\{k_j\}_{j=1}^\infty$, is
modeled by a Jacobi operator $J$ associated with the Jacobi matrix
\begin{equation}
  \label{eq:jm-0}
  \begin{pmatrix}
    q_1 & b_1 & 0  &  0  &  \cdots
\\[1mm] b_1 & q_2 & b_2 & 0 & \cdots \\[1mm]  0  &  b_2  & q_3  &
b_3 &  \\
0 & 0 & b_3 & q_4 & \ddots\\ \vdots & \vdots &  & \ddots
& \ddots
  \end{pmatrix}\,.
\end{equation}
where
\begin{equation}
\label{eq:mass-spring-matrix-entries}
q_j = -\frac{k_{j+1}+k_j}{m_j}\,, \qquad
b_j=\frac{k_{j+1}}{\sqrt{m_j m_{j+1}}}\,,
\qquad j\in\mathbb{N}\,.
\end{equation}
In solid state physics, the mass-spring system of Fig. 1 is used as a
model of one-dimensional infinite harmonic crystals (see
\cite[p.\,22]{MR1711536}). A finite mass-spring system can be used to
study molecular vibrations, where the chemical bounds between atoms
(masses) are modeled by springs \cite{smith1998infrared}.

Assuming that the movement of the system takes place within the regime
of validity of the Hooke law, one derives a Jacobi operator with
entries given by (\ref{eq:mass-spring-matrix-entries}) from the
dynamics equations (cf. \cite{MR2102477,mono-marchenko} for the finite
case). If the spectrum of $J$ is discrete, the movement of the system
is the superposition of harmonic oscillations whose frequencies are
the square root of the eigenvalues' absolute values.

In our two spectra inverse problem, one wants to find the matrix
entries corresponding to operator $J$ from the spectra of $J$ and a
perturbation of it. The perturbed operator, denoted
$\widetilde{J}_n$, has (\ref{eq:jm-1}) as its matrix representation
and corresponds to the linear semi-infinite mass-spring system given
in Fig.~2
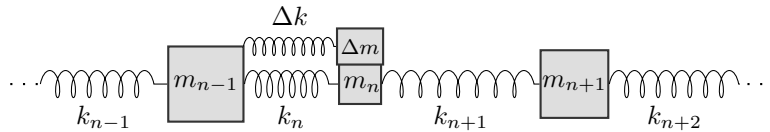
\begin{figure}[h]
\begin{center}
\begin{tikzpicture}
  [mass1/.style={rectangle,draw=black!80,fill=black!13,thick,inner sep=0pt,
   minimum size=10mm},
   mass2/.style={rectangle,draw=black!80,fill=black!13,thick,inner sep=0pt,
   minimum size=5.3mm},
   mass3/.style={rectangle,draw=black!80,fill=black!13,thick,inner sep=0pt,
   minimum size=9mm},
   dmass/.style={rectangle,draw=black!80,fill=black!13,thick,inner sep=0pt,
   minimum size=4.8mm},
   wall/.style={postaction={draw,decorate,decoration={border,angle=-45,
   amplitude=0.3cm,segment length=1.5mm}}}]
  \node (mass3) at (7.5,1) [mass3] {\footnotesize$m_{n+1}$};
  \node (mass2) at (4.65,1) [mass2] {\footnotesize$\,m_n$};
  \node (mass1) at (2.6,1) [mass1] {\footnotesize$m_{n-1}$};
  \node (dmass) at (4.65,1.5) [dmass] {\scriptsize$\,\Delta m\,$};
\draw[line width=.8pt,loosely dotted] (0,1) -- (0.4,1);
\draw[decorate,decoration={coil,aspect=0.4,segment
  length=1.3mm,amplitude=1.2mm}] (3.1,1.5) -- node[above=4pt]
{\footnotesize$\Delta k$} (dmass);
\draw[decorate,decoration={coil,aspect=0.4,segment
  length=2.1mm,amplitude=1.8mm}] (0.4,1) -- node[below=4pt]
{\footnotesize$k_{n-1}$}  (mass1);
\draw[decorate,decoration={coil,aspect=0.4,segment
  length=1.5mm,amplitude=1.8mm}] (mass1) -- node[below=4pt]
{\footnotesize$k_{n}$} (mass2);
\draw[decorate,decoration={coil,aspect=0.4,segment
  length=2.5mm,amplitude=1.8mm}] (mass2) -- node[below=4pt]
{\footnotesize$k_{n+1}$} (mass3);
\draw[decorate,decoration={coil,aspect=0.4,segment
  length=2.1mm,amplitude=1.8mm}] (mass3) -- node[below=4pt]
{\footnotesize$k_{n+2}$} (9.7,1);
\draw[line width=.8pt,loosely dotted] (9.8,1) -- (10.2,1);
%\draw[line width=.5pt,wall](0,2.1)--(0,0.7);
\end{tikzpicture}
\end{center}
\caption{Perturbed semi-infinite mass-spring system ($n\ge 2$)}\label{fig:2}
\end{figure}.

Inverse spectral problems for Jacobi operators have been amply studied
(see for instance
\cite{MR2263317,MR504044,MR2915295,MR1616422,MR0447294,MR0213379,
  MR0382314,MR549425,MR1463594,MR1436689,MR1247178} for the finite
case and \cite{MR2998707,MR1045318,MR1616422,MR499269,MR0221315,
  MR2305710,MR2438732} for the infinite case). However, inverse
spectral problems that involve the kind of perturbation producing
$\widetilde{J}_n$ from $J$ have been treated, to the best of our
knowledge, only in the finite case
\cite{MR2915295,MR1463594,MR1436689}. Yet, this sort of perturbation
arises in a natural way from the view point of physics: it corresponds
to the modification of one mass and spring constant at any place in
the chain. Noteworthily, by solving our inverse problem, we recover
the masses and spring constants of the system and the parameters of
the perturbation from the knowledge of the natural frequencies of
vibration of the original system and the perturbed one. In particular,
in the finite case, solving the inverse problem allows measuring
micromasses with the help of microcantilevers
\cite{spletzer-et-al1,spletzer-et-al2}.

To tackle the inverse problem, we use the characterization of the
relative distribution of the spectra of $J$ and $\widetilde{J}_n$
given in \cite{MR3377115}. Here a central role is played by the Green
functions of the original and perturbed operators. The Green function
is the diagonal entry of the matrix of the resolvent operator at the
point corresponding to the place in the chain where the perturbation
occurs.  Since the Green functions are fundamental for our direct and
inverse spectral analysis, we give necessary and sufficient conditions
for a meromorphic Herglotz function to be a Green function of a Jacobi
operator with discrete spectrum
(Proposition~\ref{prop:nec-suf-green}). Direct spectral analysis of
the operator and its perturbation gives a point on the real line which
is determined by the perturbation parameters and seems to act as an
``attractor'' for the eigenvalues as they are perturbed. This is
relevant for choosing proper enumerations of the set of
eigenvalues. An important conclusion of the spectral analysis is
Theorem~\ref{thm:absolute-convergence} on the convergence of the sum
of the difference of eigenvalues (cf. \cite{MR900507}).

Having  solved the direct spectral problems, we turn to solving the
conditional inverse problem. We determine the input spectral data
needed for the reconstruction of the Green function
(Proposition~\ref{prop:G-reconstruction}). Moreover, in
Theorems~\ref{thm:reconstruction-from-two-spectra} and
\ref{thm:reconstruction-from-two-spectra1}, we characterize the set of
Jacobi operators that share the same Green function and the solutions
of the two spectra inverse problem. An important ingredient for this
result is Proposition~\ref{prop:Gkk-formula} which is the key to the
theory of interior perturbations developed in
\cite{MR1616422}. Finally, we give necessary and sufficient conditions
for two sequences to be the spectra of a Jacobi operator $J$ and its
perturbation $\widetilde{J}_n$
(Theorems~\ref{thm:necessary-sufficient} and
\ref{thm:necessary-sufficient-intersection}).

This paper is a continuation of recent work on the matter
\cite{MR2915295,MR2998707,MR3113459,MR3377115} and presents
substantial generalizations of previous results. We are now able to
manage the situation where the perturbation takes place at any
arbitrary interior mass and spring of the system.  In the course of
obtaining these generalizations, unexpected nuances appeared, so it
was necessary to recur to results not needed before and develop new
techniques. It is remarkable that in the solution of the concrete
problem we have posed, various crucial problems of modern analysis
converge: the moment problem, the subtle problem of density of
polynomials in $L_2$ spaces, and various aspects in the theory of
functions.

The paper is organized as follows. In the next section, we introduce
the Jacobi operators and the finite-rank perturbation performed on
them. Some preparatory facts on Jacobi operators and their Weyl
$m$-functions are accounted for in this section. In Section
\ref{sec:green-functions}, the Green functions are defined and a
crucial formula is brought in. Here we state the necessary and
sufficient conditions for a meromorphic Herglotz function to be the
Green function of an operator with discrete spectrum. In Section
\ref{sec:direct-spectral-analysis-general-case}, the key formula
(\ref{eq:master}) is considered and results are given which describe how
the eigenvalues of the perturbed operator depend on the perturbation
parameters. Section~\ref{sec:inverse-spectral-analysis} provides
necessary and sufficient conditions on two sequences of points to be
eigenvalues of an operator $J$ and a perturbation of it. Finally, in
the Appendix, we include a result on the representation of Weyl
$m$-functions of Jacobi operators on the basis of a classical result due
to M.\,G. Krein.

\section{Jacobi operators}
\label{sec:jacobi_operators}

For a sequence $f=\{f_k\}_{k=1}^\infty$ of complex numbers, consider the
second order difference expressions
\begin{subequations}
  \label{eq:difference-expr}
\begin{align}
   \label{eq:difference-recurrence}
  (\Upsilon f)_k&:= b_{k-1}f_{k-1} + q_k f_k + b_kf_{k+1}
  \quad k \in \mathbb{N} \setminus \{1\},\\
   \label{eq:difference-initial}
   (\Upsilon f)_1&:=q_1 f_1 + b_1 f_2\,,
\end{align}
\end{subequations}
where $q_k\in\mathbb{R}$ and $b_k>0$ for any $k\in\mathbb{N}$. We
remark that (\ref{eq:difference-initial}) can be seen as a boundary
condition.

Let $l_2(\mathbb{N})$ be the space of square summable complex
sequences. In this Hilbert space, define the operator $J_0$ whose
domain contains only the sequences having a finite number of non-zero
elements and is given by $J_0f:=\Upsilon f$.  Clearly, the operator
$J_0$ is symmetric and therefore closable, so one can consider the
operator $\cc{J_0}$ being its closure. It turns out that
$\overline{J_0}$ is the operator whose matrix representation with
respect to the canonical basis $\{\delta_n\}_{n=1}^\infty$ in
$l_2(\mathbb{N})$ is
\begin{equation}
  \label{eq:jm-0}
  \begin{pmatrix}
    q_1 & b_1 & 0  &  0  &  \cdots
\\[1mm] b_1 & q_2 & b_2 & 0 & \cdots \\[1mm]  0  &  b_2  & q_3  &
b_3 &  \\
0 & 0 & b_3 & q_4 & \ddots\\ \vdots & \vdots &  & \ddots
& \ddots
  \end{pmatrix}\,.
\end{equation}
See \cite[Sec. 47]{MR1255973} for the definition of the matrix
representation of an unbounded symmetric operator.

One of the following two possibilities for the deficiency indices
of $\cc{J_0}$ holds \cite[Chap.\,4,\,Sec.\,1.2]{MR0184042}:
\begin{subequations}
  \label{eq:deficiency-indices}
  \begin{align}
     n_+(\cc{J_0})&=n_-(\cc{J_0})=1\,,\label{eq:one-one}\\
     n_+(\cc{J_0})&=n_-(\cc{J_0})=0\,.\label{eq:nil-nil}
  \end{align}
\end{subequations}
Let $J$ be a self-adjoint extension of $J_0$. Thus, in view of
(\ref{eq:deficiency-indices}), the von Neumann extension theory tells
us that either $J$ is a proper closed symmetric extension of
$\cc{J_0}$ or $J=\overline{J_0}$. In the general case, there are
various operators $J$ associated with the matrix \eqref{eq:jm-0} and
we referred to them generically as Jacobi operators associated with
\eqref{eq:jm-0}.

Within the regime of validity of the Hooke law, the Jacobi operator
$J$ models the semi-infinite linear mass-spring system of Fig.~1
\cite{MR2998707,MR3377115} with
(\ref{eq:mass-spring-matrix-entries}).  See
\cite{MR2102477,mono-marchenko} for an explanation of the deduction of
these formulae in the finite case.

Fix $n\in\nats$ and consider, along with the self-adjoint operator
$J$, the operator
\begin{equation}
\label{eq:def-tilde-j}
\begin{split}
  \widetilde{J}_n=J &+
  [q_n(\theta^2-1)+\theta^2h]\inner{\delta_n}{\cdot}\delta_n \\
  &+  b_n(\theta-1)(\inner{\delta_n}{\cdot}\delta_{n+1} +
  \inner{\delta_{n+1}}{\cdot}\delta_n) \\
  &+ b_{n-1}(\theta-1)(\inner{\delta_{n-1}}{\cdot}\delta_{n} +
  \inner{\delta_{n}}{\cdot}\delta_{n-1})
  \,,\quad \theta>0\,,
  \quad h\in\mathbb{R}\,,
\end{split}
\end{equation}
where it has been assumed that $b_0=0$. Clearly, $\widetilde{J}_n$ is a
self-adjoint extension of the operator whose matrix representation
with respect to the canonical basis in $l_2(\mathbb{N})$ is a Jacobi
matrix obtained from (\ref{eq:jm-0}) by modifying the entries
$b_{n-1},q_n,b_n$. For instance, if $n>2$, $\widetilde{J}_n$ is a
self-adjoint extension (possibly not proper) of the operator whose
matrix representation is
\begin{equation}
  \label{eq:jm-1}
  \begin{pmatrix}
q_1 & b_1 & 0 & 0 & 0 & 0 & \cdots \\[1mm]
b_1 & \ddots & \ddots & 0 & 0 & 0 & \cdots \\[1mm]
0  &  \ddots  & q_{n-1} & \theta b_{n-1} & 0 & 0 & \cdots\\
0 & 0 & \theta b_{n-1} & \theta^2(q_n+h) & \theta b_n & 0 & \cdots\\
0 & 0 & 0 & \theta b_{n} & q_{n+1} & \theta b_{n+1} & \\
0 & 0 & 0 & 0 & b_{n+1} & q_{n+2} & \ddots\\
\vdots & \vdots & \vdots & \vdots & & \ddots & \ddots
  \end{pmatrix}\,.
\end{equation}
Note that $\widetilde{J}_n$ is obtained from $J$ by a rank-three
perturbation when $n>1$, and a rank-two perturbation otherwise.

The operator $\widetilde{J}_n$ serves as a model of the perturbed mass
spring system of Fig.~2, where
\begin{equation}
\label{eq:increment-mass-spring}
\Delta m=m_n(\theta^{-2}-1)\quad\text{ and }\quad\Delta
k=-hm_n.
\end{equation}
By setting $f_1=1$, a solution of the equations
\begin{subequations}
\label{eq:spectral-equation}
\begin{align}
  \label{eq:initial-spectral}
  (\Upsilon f)_1&:= zf_1\,,\\
  \label{eq:recurrence-spectral}
  (\Upsilon f)_k&:= z f_k\,,
  \quad k \in \mathbb{N} \setminus \{1\},
\end{align}
\end{subequations}
can be found uniquely by recurrence. This solution, denoted by
$\pi(z)=\{\pi_k(z)\}$, is such that $\pi_k(z)$ is a polynomial of
degree $k-1$. Alongside this sequence, define the sequence $\xi(z)$
as the solution of (\ref{eq:recurrence-spectral}) after setting
$f_1=0$ and $f_2=b_1^{-1}$. Thus, $\xi_k(z)$ is a polynomial of
degree $k-2$. The elements of the sequence $\pi(z)$, respectively
$\xi(z)$, are referred to as the polynomials of the first,
respectively second, kind associated with the matrix
(\ref{eq:jm-0}). By comparing (\ref{eq:difference-expr}) with
(\ref{eq:spectral-equation}), one concludes that for $\pi(z)$ to be in
$\ker(J_0^*-zI)$, it is necessary and sufficient that $\pi(z)$ be an
element of $l_2(\nats)$. Of course, $\pi(z)\in\ker(J-zI)$, if and only
if $\pi(z)\in\dom(J)$.

It follows from the definition of the operator $J$ that
\begin{equation}
  \label{eq:delta-k-through-delta-1}
  \delta_k=\pi_k(J)\delta_1\quad\forall k\in\nats\,.
\end{equation}
This implies that $J$ is simple and $\delta_1$ is a cyclic vector (see
\cite[Sec. 69]{MR1255973}). Therefore, by defining
\begin{equation}
  \label{eq:spectral-function-def}
  \rho(t):=\inner{\delta_1}{E(t)\delta_1}\,,\quad t\in\reals\,,
\end{equation}
where $E$ is the resolution of the identity given by the spectral
theorem, one has, due to \cite[Sec. 69, Thm. 2]{MR1255973}), that
there is a unitary map $\Phi:L_2(\reals,\rho)\to l_2(\nats)$ such that
$\Phi^{-1}J\Phi$ is the multiplication by the independent variable
defined in its maximal domain. We call the function given by
\eqref{eq:spectral-function-def} the spectral function of the Jacobi
operator $J$.  Moreover, due to \cite[Sec. 69, Thm. 2]{MR1255973}), it
follows from (\ref{eq:delta-k-through-delta-1}) that the function
$\pi_k\eval{\reals}$ belongs to $L_2(\reals,\rho)$ for all
$k\in\nats$, i.\,e., all moments of $\rho$ exists (see also
\cite[Thm.\,4.1.3]{MR0184042}). The equation
\eqref{eq:delta-k-through-delta-1} means that
\begin{equation}
  \label{eq:unitary-map}
  \Phi\pi_k=\delta_k\,,\quad\forall k\in\nats\,,
\end{equation}
which implies that the polynomials are dense in $L_2(\reals,\rho)$
since $\Phi$ is unitary. Thus, the discussion above has proven one
direction of the following assertion which follows directly from
\cite[Thms.\,2.3.3 and 4.1.4]{MR0184042}.
% It turns out
% that the following holds
\begin{theorem}
  \label{thm:nec-suf-cond-for-meausure}
 A function $\rho$ is the spectral function of a Jacobi operator if
 and only if $\int_\reals d\rho=1$, all the polynomials are in
  $L_2(\reals,\rho)$ and they are dense in this space.
\end{theorem}
% \begin{proof}
%   ($\Rightarrow$) This follows from the discussion
%   above. ($\Leftarrow$). By \cite[Thms. 2.3.3,\,4.1.4]{MR0184042}, one knows
%   that the polynomials are dense in  $L_2(\reals,\rho)$
% \end{proof}
\begin{remark}
  \label{rem:nec-suf-conditions-for-measures}
  Recall that $J$ is either a proper self-adjoint
  restriction of $J_0^*$ or $J=\cc{J_0}$ depending on the alternative
  given in \eqref{eq:deficiency-indices}. Thus, for any measure $\rho$ such that all
  polynomials are a dense linear subset of $L_2(\reals,\rho)$, there
  exists a Jacobi operator $J$, viz. a canonical self-adjoint
  extension of the operator whose matrix representation is a Jacobi
  matrix, such that $\rho$ and $J$ are related by
  \eqref{eq:spectral-function-def}.
\end{remark}

\begin{remark}
  \label{rem:finite-matrix}
  Any measure with finite support is the spectral measure of the
  operator associated with some finite Jacobi matrix.
\end{remark}

\begin{definition}
  \label{def:weyl-function}
The Weyl $m$-function is defined as follows
\begin{equation}
  \label{eq:weyl-function}
  m(z):=\inner{\delta_1}{(J-z I)^{-1}\delta_1}\,,\qquad z\not\in\sigma(J)\,,
\end{equation}
where $\sigma(J)$ denotes the spectrum of $J$.
\end{definition}
By the map $\Phi$, it immediately follows from this
definition that
\begin{equation}
  \label{eq:weyl-by-spectral-th}
  m(z)=\int_{\mathbb{R}}
  \frac{d\rho(t)}{t-z}\,.
\end{equation}
Thus, by the Nevanlinna representation theorem (see
\cite[Thm.\,5.3]{MR1307384}), $m(z)$ is a Herglotz function.  Recall
that a \textbf{Herglotz} function  $f$ (also called Pick or Nevanlinna-Pick
function) is holomorphic in the upper half plane and
\begin{equation*}
\im f(z)\ge
0\quad\text{ whenever }\im z>0.
\end{equation*}

\begin{remark}
\label{rem:inverseJM}
For Jacobi operators, the inverse spectral theory is based on the fact
that, from the Weyl $m$-function (or, equivalently, $\rho$), one
uniquely recovers the matrix (\ref{eq:jm-0}) and the boundary
condition at infinity that defines the self-adjoint extension if
necessary. This is done by means of either the discrete Riccati
equation (see \cite[Eq.\,2.15]{MR1616422},
\cite[Eq.\,2.23]{MR1643529}) or the method of orthonormalization of
the polynomial sequence $\{t^k\}_{k=0}^\infty$ in
$L_2(\mathbb{R},\rho)$ \cite[Chap.\,7,\,Sec.\,1.5]{MR0222718}. If
(\ref{eq:jm-0}) is the matrix representation of a non-self-adjoint
operator, then the condition at infinity may be found by the method
exposed in \cite[Sec.\,2]{MR2305710}.
\end{remark}
\section{Green functions for Jacobi operators}
\label{sec:green-functions}

We begin this section by introducing some concepts and laying out the
notation.

Let $M\subset\integers$ and consider a sequence $\{M_n\}_{n=1}^\infty$ of
finite sets
containing only consecutive integers and such that
\begin{enumerate}[i)]
%\item $M_n\subset M$ for any $n\in\nats$,
\item $M_n\subset M_{n+1}$ for any $n\in\nats$,
\item $\cup_{n=1}^\infty M_n=M$\,.
\end{enumerate}
If, for a collection of complex numbers $\{r_k\}_{k\in M}$,
\begin{equation*}
  \lim_{n\to\infty}\sum_{k\in M_n}r_k=s\,,
\end{equation*}
for any sequence $\{M_n\}_{n=1}^\infty$ satisfying i) and  ii), then we define
\begin{equation*}
  \sum_{k\in M}r_k:=s
\end{equation*}
and  say that $\sum_{k\in M}r_k$ converges to $s$. For any finite
subset $F$ of $M$, we also consider
\begin{equation*}
  \sum_{k\in M\setminus F}r_k:=\sum_{k\in M}r_k-\sum_{k\in F}r_k\,.
\end{equation*}

The expressions
\begin{equation*}
  \prod_{k\in M}r_k \quad\text{ and }\quad\prod_{k\in M\setminus F}r_k
\end{equation*}
are defined analogously.

For every $z\in\complex\setminus\rho(J)$, define
\begin{equation}
  \label{eq:psi-vector}
  \psi(z):=(J-zI)^{-1}\delta_1\,.
\end{equation}
It is known that \cite[Chap.\,7 Eq.\,1.39]{MR0222718} that for each
$z\in\complex\setminus\rho(J)$ there exists a unique complex number
$m(z)$ such that
\begin{equation}
  \label{eq:weyl-solution}
  \psi(z)=m(z)\pi(z) +\xi(z)\,.
\end{equation}
The notation here corresponds to the fact that the number $m(z)$ is
actually the value of the Weyl $m$-function at $z$.

\begin{definition}
  \label{def:submatrices}
  Given a subspace $\mathcal{G}\subset\l_2(\nats)$, let
  $P_{\mathcal{G}}$ be the orthogonal projection onto
  $\mathcal{G}$ and
  $\mathcal{G}^\perp:=\{\phi\in
  l_2(\nats):\inner{\phi}{\psi}=0\
  \forall\,\psi\in\mathcal{G}\}$. Also, define
  the subspace
$\mathcal{F}_n:=\Span\{\delta_k\}_{k=1}^n$. In the Hilbert space $\mathcal{F}_n^\perp$,
consider the operator
\begin{equation*}
  J_n^+:=P_{\mathcal{F}_n^\perp}J\eval{\mathcal{F}_n^\perp}\,,\quad n\in\nats\setminus\{1\}\,.
\end{equation*}
Similarly, in the space $\mathcal{F}_{n-1}$, consider
  \begin{equation*}
    J_n^-:=P_{\mathcal{F}_{n-1}}J\eval{\mathcal{F}_{n-1}}\,,\quad n\in\nats\setminus\{1\}\,.
  \end{equation*}
Here, we have used the notation
 $J\eval{\mathcal{G}}$ for the restriction of $J$ to the set
 $\mathcal{G}$, that is,
 $\dom(J\eval{\mathcal{G}})=\dom(J)\cap\mathcal{G}$. The corresponding
 Weyl $m$-functions of these operators are
 \begin{equation*}
   m_n^+(z):=\inner{\delta_{n+1}}{(J_n^+-zI)^{-1}\delta_{n+1}}\,,\qquad
    m_n^-(z):=\inner{\delta_{n-1}}{(J_n^--zI)^{-1}\delta_{n-1}}\,.
 \end{equation*}
\end{definition}

The operator $J_n^+$ is a self-adjoint extension of the operator whose
matrix representation with respect to the basis
$\{\delta_k\}_{k=n}^\infty$ of the space
$\mathcal{F}_n^\perp$ is (\ref{eq:jm-0}) with
the first $n$ rows and $n$ columns removed. When $J_0$ is not
essentially self-adjoint, $J_n^+$ has the same boundary conditions at
infinity as the operator $J$. Clearly, the operator $J_n^-$ lives in
an $(n-1)$-dimensional space.
\begin{definition}
\label{def:green-function}
For any $n\in\nats$, we use the following notation
\begin{equation*}
  G(z,n):=\inner{\delta_n}{(J-zI)^{-1}\delta_n}\,,\qquad z\not\in\sigma(J)\,,
\end{equation*}
and call $G(z,n)$ the $n$-th Green function of the Jacobi operator $J$.
Observe that $G(z,1)=m(z)$ (See Definition~\ref{def:weyl-function}).
\end{definition}
In view of
(\ref{eq:delta-k-through-delta-1}) and \eqref{eq:spectral-function-def}, one has
\begin{equation}
  \label{eq:green-integral}
  G(z,n)=\int_\reals\frac{\pi_n^2(t)d\rho(t)}{t-z}\,,\qquad z\not\in\sigma(J)\,.
\end{equation}
Thus, for any $n\in\nats$, $G(\cdot,n)$ is a Herglotz function. This
function is extended analytically to the eigenvalues of $J$ which are zeros
of $\pi_n$ since these points are removable singularities.

On the basis of the
von Neumann expansion for the resolvent
(cf. \cite[Chap.\,6,\,Sec.\,6.1]{MR1711536}), one has
\begin{equation*}
  (J-zI)^{-1}\delta_n=
  -\sum_{k=0}^{N-1}\frac{J^k}{z^{k+1}}\delta_n
  +\frac{J^N}{z^{N}}
  (J-z I)^{-1}\delta_n
\end{equation*}
for any $n\in\nats$ and $z\in\mathbb{C}\setminus\sigma(J)$. From this and the fact that (see
\cite[Chap.\,6 Sec.\,3]{MR1192782})
\begin{equation*}
  \norm{(J-z I)^{-1}}\le\frac{1}{\abs{\im z}}\,,
\end{equation*}
one can  obtain the following asymptotic formula
\begin{equation}
  \label{eq:G=-asymptotics}
  G(z,n)=-\frac{1}{z} +O(z^{-2})
\end{equation}
as $z\to\infty$ along any curve away from the spectrum.

The next assertion is proven in \cite[Thm. 2.8]{MR1616422} and
\cite[Prop.\,2.3]{MR3377115}.
\begin{proposition}
  \label{prop:Gkk-formula}
  For any $n\in\nats$
\begin{equation}
  \label{eq:Gkk-formula2}
   G(z,n)=\frac{-1}{b_n^2m_n^+(z)+b_{n-1}^2m_n^-(z)+z-q_n}\,,
\end{equation}
where we define $m^-_1(z)\equiv 0$.
\end{proposition}

The case $n=1$ in \eqref{eq:Gkk-formula2} is the
Riccati type equation for the Weyl $m$-function
\cite[Eq.\,2.15]{MR1616422}, \cite[Eq.\,2.23]{MR1643529}.

Note that, in the case when $J_0$ is not
essentially self-adjoint, the dependence of $G$ on the choice of the
self-adjoint extension is given by the fact that $m_n^+(z)$ depends on
this extension.

\begin{hypothesis}
 \label{hyp:discreteness}
  Hypothesis: The Jacobi operator $J$ has discrete spectrum, that is,
  \begin{equation*}
    \sigma_{ess}(J)=\emptyset\,.
  \end{equation*}
  The essential spectrum $ \sigma_{ess}(J)$,
  in this case, is the accumulation points of $\sigma(J)$.
  Since $\widetilde{J}_n$
  and $J$
  differ by a finite-rank perturbation, it follows from the Weyl
perturbation theorem that the spectrum of
  $\widetilde{J}_n$ is also discrete. For the same
  reason, $J_n^+$ has also discrete spectrum.
\end{hypothesis}
% \begin{remark}
%   \label{rem:jn-also-discrete}
% Since the operator $J$ can be seen as a finite-rank perturbation of
% the orthogonal sum of $J_n^+$ and $J_n^-$, it follows by Weyl
% perturbation theorem that $J_n^+$ has also discrete spectrum.
% \end{remark}

\begin{remark}
  \label{rem:zeros-poles-hm}
  If one assumes H\,1, then the functions $m(z)$
  and $G(z,n)$
  are meromorphic. A consequence of being Herglotz and meromorphic, is
  that these functions have real and simple zeros and poles.
  Moreover, zeros interlace with poles, that is between two contiguous
  zeros there is exactly one pole and between two contiguous poles
  there is only one zero \cite[Chap.\,7, Thm.\,1]{MR589888}.
\end{remark}
\begin{proposition}
  \label{prop:form-inverse-green-function}
Assume H\,1 and let $\{\alpha_k\}$ be the zeros of the Green function $G(z,n)$ for any
$n\in\nats$. Then
there are constants $\eta_k\ge 0$ ($k\in\nats$) such that
  \begin{equation*}
    -G(z,n)^{-1}=z-q_n+\sum_{k\in M}\frac{\eta_k}{\alpha_k-z}\,,
  \end{equation*}
where $q_n$ is the $n$-th element of the main diagonal of (\ref{eq:jm-0}).
\end{proposition}
\begin{proof}
By
(\ref{eq:Gkk-formula2}), one has
\begin{align}
\label{eq:decomposition-of-Gkk-formula2}
  -G(z,n)^{-1}&=z-q_n+b_n^2m_n^+(z)+b_{n-1}^2m_n^-(z)\nonumber\\
  &=z-q_n+b_n^2\sum_{k\in\widetilde{M}}\frac{\tau_k}{c_k-z}+
b_{n-1}^2\sum_{j=1}^{n-1}\frac{\kappa_j}{d_j-z}\,,
\end{align}
where in the last equality we have used (\ref{eq:weyl-by-spectral-th})
and the fact that, due to Hypothesis~\ref{hyp:discreteness}, $J_n^+$
has discrete spectrum. Clearly, the set of the union
of the elements of $\{c_k\}_{k\in\widetilde{M}}$ and $\{d_k\}_{k=1}^{n-1}$
are the zeros of $G(z,n)$.
\end{proof}
\begin{remark}
  \label{rem:chebotarev}
  Since $G(z,n)$ is a Herglotz meromorphic function when H\,1 is
  assumed, the same is true for the function $-G(z,n)^{-1}$. Using
  \cite[Chap.\,7, Thm.\,2]{MR589888}, one writes
  \begin{equation}
    \label{eq:chebotarev}
    -G(z,n)^{-1}=az+b+\sum_{k\in M}\eta_k\left(\frac{1}{\alpha_k-z}
 -\frac{1}{\alpha_k}\right)\,,
  \end{equation}
where $a\ge 0$, $b$ is real and $\eta_k\ge 0$ for all $k\in\nats$.
Comparing this last equation with the statement of
Proposition~\ref{prop:form-inverse-green-function}  one concludes
that
\begin{equation}
  \label{eq:moment-1-convergent}
  \sum_{k\in M}\frac{\eta_k}{\alpha_k}<+\infty\,.
\end{equation}
Actually, it follows from \eqref{eq:Gkk-formula2} that
for an infinite set of subindices $\eta_k=b_n^2\tau_k$ (see
(\ref{eq:decomposition-of-Gkk-formula2})), therefore, as a consequence
of Theorem~\ref{thm:nec-suf-cond-for-meausure},
\begin{equation*}
  \sum_{k\in M}\alpha_k^m\eta_k<\infty\,,\quad\text{ for all }m=0,1,\dots
\end{equation*}
since $\tau_k$ is jump at $\alpha_k$ of the spectral measure of the
infinite submatrix whose Weyl $m$-function is $m^+$.
\end{remark}

\begin{proposition}
  \label{prop:nec-suf-green}
  Let $G(z)$ be a meromorphic function.  Denote
  by $-\eta_k$ the residue of $-G(z)^{-1}$ at $\alpha_k$. $G(z)$ is the $n$-th
  Green function of a Jacobi operator satisfying H\,1 for some
  $n\in\nats\setminus\{1\}$ if and only if all the following
  conditions are satisfied
  \begin{enumerate}
  \item  $G(z)$ is a Herglotz function.
  \item $G(z)$ obeys the asymptotics
    \begin{equation*}
      G(z)=-\frac{1}{z}+O(z^{-2})
    \end{equation*}
when $z$ tends to $\infty$ along any curve away from a strip containing the real line.
  \item There exists a finite set $F$ with $\card F=n$ such that all the polynomials are in
  $L_2(\reals,\rho)$ and they are dense in this space, where
  \begin{equation}
    \label{eq:measure-infinite-submatrix}
    \rho(t):=\sum_{\substack{\alpha_k<t\\ k\in \nats\setminus F}}\eta_k\,.
  \end{equation}
  \end{enumerate}
\end{proposition}
\begin{proof}
  ($\Rightarrow$) Condition 1 follows from
  \eqref{eq:green-integral} and \eqref{eq:G=-asymptotics} implies
  Condition 2. It follows from Propositions~\ref{prop:Gkk-formula} and
  \ref{prop:form-inverse-green-function} that
  \begin{equation*}
   \sum_{k\in M}\frac{\eta_k}{\alpha_k-z}=\left(\sum_{k\in M\setminus
     F}+\sum_{k\in F}\right)\frac{\eta_k}{\alpha_k-z}\,,
  \end{equation*}
where the infinite sum in the r.\,h.\,s. of the last equation is the
Weyl $m$-function of a semi-infinite Jacobi operator. Finally one recurs
to Theorem~\ref{thm:nec-suf-cond-for-meausure} to show that Condition
3 holds.

($\Leftarrow$) Since $-G(z)^{-1}$ is also Herglotz, one has
\eqref{eq:chebotarev}. Condition 3 implies in particular that
\begin{equation*}
  \sum_{k\in M}\eta_k<+\infty\,.
\end{equation*}
Thus, \eqref{eq:moment-1-convergent} holds and then one can write
\begin{equation*}
  -G(z,n)^{-1}=az+\widetilde{b}+\sum_{k\in M}\frac{\eta_k}{\alpha_k-z}\,.
\end{equation*}
Because of Condition 2, $a=1$. Take the set $F$ given in  Condition 3
and write
\begin{equation*}
  -G(z,n)^{-1}=z+\widetilde{b}+\left(\sum_{k\in M\setminus
    F}+\sum_{k\in F}\right)\frac{\eta_k}{\alpha_k-z} \,,
\end{equation*}
where, due to Theorem~\ref{thm:nec-suf-cond-for-meausure} and
Condition 3,
\begin{equation*}
  \sum_{k\in M\setminus F}\frac{\eta_k}{\alpha_k-z}
\end{equation*}
is the Weyl $m$-function (modulo a positive constant factor) of some
semi-infinite Jacobi operator. On the other hand,
\begin{equation*}
  \sum_{k\in F}\frac{\eta_k}{\alpha_k-z}
\end{equation*}
is the Weyl $m$-function (modulo a positive constant factor) of a
finite Jacobi matrix.
Finally,
Proposition~\ref{prop:Gkk-formula} completes the
proof.
\end{proof}
\section{Direct spectral analysis}
\label{sec:direct-spectral-analysis-general-case}

We begin this section by defining the function
\begin{equation*}
  \widetilde{G}(z,k):=\inner{\delta_k}{(\widetilde{J}_n-zI)^{-1}\delta_k}
\end{equation*}
to be considered alongside the function $G(z,k)$. Define
\begin{equation}
  \label{eq:M-definition}
  \mathfrak{M}_k(z):=\frac{G(z,k)}{\widetilde{G}(z,k)}\,.
\end{equation}
This function is extended to the points that are removable singularities.

The following formula plays an important role for the comparative
spectral analysis of $J$ and $\widetilde{J}_n$
\cite[Lem.\,3.1]{MR3377115}.
\begin{equation}
  \label{eq:master}
  \mathfrak{M}_k(z)=\theta^2+(1-\theta^2)(\gamma -z)G(z,k)\,,
\end{equation}
where
\begin{equation}
  \label{eq:def-gamma}
  \gamma:=\frac{\theta^2h}{1-\theta^2}\,.
\end{equation}
This formula follows from (\ref{eq:Gkk-formula2}).

% By (\ref{eq:master}) and \cite[Theorem~A.1]{2013arXiv1312.6920D}, one
% infers that the function $\mathfrak{M}$ behaves as illustrated in
% Fig.~%\ref{fig:3}
% (cf. \cite[Thms.\,3.1 and 3.2]{2013arXiv1312.6920D}).
% \begin{figure}[h]
% \begin{center}
% \vspace{6cm}
% \end{center}
% \caption{Behaviour of $\mathfrak{M}$ on the real axis}\label{fig:3}
% \end{figure}
The following assertion reproduces the one of
\cite[Cor. 3.2]{MR3377115}. This result can be seen by comparing the
formula \eqref{eq:Gkk-formula2} for $G(z,n)$ and $\widetilde{G}(z,n)$.
\begin{proposition}
  \label{prop:zero-green-zero-perturbed}
  For any $n\in\nats$, the function $G(z,n)$ vanishes if and only if
  $\widetilde{G}(z,n)$ vanishes.
\end{proposition}
% \begin{proof}
%   Prove it completely.The case $z\ne\gamma$ has been proven in
%   \cite[Cor. 3.1]{MR3377115} in a more general setting. When
%   $z=\gamma$, due to the discreteness of the spectra of $J$ and
%   $\widetilde{J}_n$, the functions $G(z,n)$ and $\widetilde{G}(z,n)$ are
%   meromorphic. The assertion, then follows from
%   (\ref{eq:M-definition}), taking into account that
%   $\mathfrak{M}_k(\gamma)=\theta^2$.
% \end{proof}
\begin{proposition}
  \label{prop:common-eigenvalues-+-}
  Let $n\in\nats\setminus\{1\}$. The set of common eigenvalues of $J$
  and $\widetilde{J}_n$, different from $\gamma$, is equal to the set
  of common eigenvalues of $J_n^+$ and $J_n^-$, different from
  $\gamma$. The function $G(z,n)$ vanishes at these eigenvalues.
\end{proposition}
\begin{proof}
  The fact that
  \begin{equation*}
    \sigma(J_n^+)\cap\sigma(J_n^-)\subset\sigma(J)\cap\sigma(\widetilde{J}_n)
  \end{equation*}
  has been proven in \cite[Prop.\,3.3]{MR3377115}. To prove the
  converse contention, suppose that $\lambda$ is in
  $\sigma(J)\cap\sigma(\widetilde{J}_n)$.  Then, by
  \cite[Lem.\,2.8]{MR3377115}, $\lambda$ is either a zero or a pole of
  $G(z,n)$. But \cite[Lem.\,3.5]{MR3377115} implies that, if
  $\lambda\ne\gamma$, then $\lambda$ cannot be a pole of $G(z,n)$.
  The fact that $\lambda$ is a zero of $G(z,n)$ and an eigenvalue of
  $J$ yields through \cite[Lems.\,2.8 and 2.9, and
  Cor.\,2.3]{MR3377115} that $\lambda$ is a common eigenvalue of
  $J_n^+$ and $J_n^-$. The last assertion is
  \cite[Lem. 3.5]{MR3377115}.
\end{proof}
% \begin{corollary}
%   \label{cor:common-eigenvalues-zero-g}
%   Let $n\in\nats\setminus\{1\}$.
%   \begin{equation*}
%     \sigma(J)\cap\sigma(\widetilde{J}_n)\subset\{r\in\reals: G(r,n)=0\}
%   \end{equation*}
% \end{corollary}
% \begin{proof}
%   The proof follows directly from the proof of the previous proposition.
% \end{proof}
\begin{remark}
  \label{rem:finite-common-eigenvalues}
  A consequence of the previous proposition is that the number of
  common eigenvalues different from $\gamma$ is not greater than
  $n-1$. Moreover these common eigenvalues are the same for any $\theta\in(0,1)$
  and $h\in\reals$.
 % there is a finite
 %  number of common eigenvalues of $J$ and $\widetilde{J}_n$.
\end{remark}

% The following assertion is proven in \cite[Thms.\,3.8, 3.10 and
% Prop.\,3.4]{MR3377115} (see also \cite{MR2915295} for the
% analogous assertion for the finite dimensional case)
% \begin{proposition}
%   \label{prop:interlacing}
%   Let $\theta<1$. Enumerate the discrete sets $\sigma(J)\cup\{\gamma\}$ and
%   $\sigma(\widetilde{J}_n)$ according to (C) and denote
%   $\sigma(J)\cup\{\gamma\}=\{\nu_k\}_{k\in M}$ and
%   $\sigma(\widetilde{J}_n)=\{\mu_k\}_{k\in M}$ with $\nu_{k_0}=\gamma$. Then
%   \begin{align*}
%     \nu_k&\le\mu_k<\nu_{k+1}\qquad k<k_0\,,\\
%     \nu_{k_0}&\le\mu_{k_0}\le \nu_{k_0+1}\\
%     \nu_k&<\mu_k\le\nu_{k+1}\qquad k>k_0
%   \end{align*}
% and
% $\card((\sigma(J)\cap\sigma(\widetilde{J}_n))\setminus\{\gamma\}\le
% n$. Moreover, if $\gamma$ is in $\sigma(J)$, then $\nu_{k_0}=\mu_{k_0}$.
% \end{proposition}

% \begin{remark}
% \label{rem:true-interlacing}
% Clearly, if two real sequences $S$, $S'$ without finite accumulation
% points interlace, then one always can find $M$ and functions $f:M\to
% S$ and $f':M\to S'$ with the properties given in our convention (C) such
% that, for any $k\in M$, either
% \begin{equation*}
%   \lambda_k<\lambda_k'<\lambda_{k+1}\quad\text{ or }\quad
% \lambda_k'<\lambda_k<\lambda_{k+1}'\,,
% \end{equation*}
% where $\lambda_k=f(k)$ and $\lambda_k'=f'(k)$. If $S$ is not
% semi-bounded, then both possibilities hold simultaneously.
% \end{remark}

\begin{notation}
  Notation: Assume H\,1 and denote $J(\theta,h):=\widetilde{J}_n$ to
  emphasize the dependence on $\theta$ and $h$ of the operator
  $\widetilde{J}_n$ (recall that $\widetilde{J}_n$ is the operator
  whose matrix representation is given by \eqref{eq:jm-1}). Denote the
  spectrum of $J(\theta,h)$ by $\{\lambda_k(\theta,h)\}_{k\in
    M}$. Note that this sequence has no accumulation points for any
  fixed $\theta>0$ and $h\in\reals$.
 Let
  \begin{equation*}
    \pi(\theta,h)=\{\pi_j(\theta,h)\}_{j=1}^\infty
  \end{equation*}
  be the eigenvector of $J(\theta,h)$ corresponding to
  $\lambda_k(\theta,h)$ normalized as before (see
  \eqref{eq:spectral-equation}), i.\,e.,
  \begin{equation*}
    \inner{\delta_1}{\pi(\theta,h)}=1\,.
  \end{equation*}
\end{notation}
Note that the polynomial $\pi_n(\theta,h)$ of degree $n-1$ is
evaluated at the point $\lambda_k(\theta,h)$.
\begin{lemma}
  \label{lem:derivative}
 Assume  H\,1 and consider N\,1. For any $k\in M$, $\theta>0$ and
 $h\in\reals$, one has
  \begin{enumerate}[a)]
  \item
    \begin{equation*}
      \frac{d}{d\theta}\lambda_k(\theta,h)=\frac{2\pi_n}{\norm{\pi}^2}
\left(b_{n-1}\pi_{n-1}+b_{n}\pi_{n+1} + \pi_{n}\theta(q_n+h)\right)\,,
    \end{equation*}
where we have abbreviated $\pi:=\pi(\theta,h)$ and $\pi_n:=\pi_n(\theta,h)$.
  \item
    \begin{equation*}
      \frac{d}{dh}\lambda_k(\theta,h)=\frac{\theta^2\pi_n^2(\theta,h)}{\norm{\pi(\theta,h)}^2}\,.
    \end{equation*}
  \end{enumerate}
\end{lemma}
\begin{proof}
  Observe that if $A,A'$ are symmetric operators with the same domain
  and $Af=\lambda f$, $A'f'=\lambda'f'$, then
  \begin{equation}
    \label{eq:symmetric-relation}
    \inner{f}{(A'-A)f'}=\inner{f}{\lambda'f'}-\inner{f}{Af'}
=\inner{f}{\lambda'f'}-\inner{Af}{f'}=(\lambda'-\lambda)\inner{f}{f'}\,.
\end{equation}
Proof of a)\\
Since $h$ will remain fixed we do not write the dependence on
$h$. Pick any small real $\tau$. Since the domain of $J(\theta,h)$
does not depend on $\theta$, using \eqref{eq:symmetric-relation}, we
have, similar to \cite[Prop.\,3.1]{MR3377115},
\begin{equation}
  \label{eq:sym-relation-for-J}
  (\lambda_k(\theta+\tau)-\lambda_k(\theta))\inner{\pi(\theta}{\pi(\theta+\tau}=
 \inner{\pi(\theta)}{(J(\theta+\tau)-J(\theta))\pi(\theta+\tau)}\,.
\end{equation}
Let us calculate the inner product of the right hand side of the above
equality.

Note that
\begin{equation}
  \label{eq:difference-theta-tau-theta}
  J(\theta+\tau)-J(\theta)=
\begin{pmatrix}
0 & 0 & 0 & 0 & 0 & 0 & \cdots \\[1mm]
0 & \ddots & \ddots & 0 & 0 & 0 & \cdots \\[1mm]
0  &  \ddots  & 0 & \tau b_{n-1} & 0 & 0 & \cdots\\
0 & 0 & \tau b_{n-1} & (2\tau\theta+\tau^2)(q_n+h) & \tau b_n & 0 & \cdots\\
0 & 0 & 0 & \tau b_{n} & 0 & 0 & \\
0 & 0 & 0 & 0 & 0 & 0 & \ddots\\
\vdots & \vdots & \vdots & \vdots & & \ddots & \ddots
  \end{pmatrix}\,.
\end{equation}
Using the notation $\pi:=\pi(\theta)$, $\pi_n:=\pi_n(\theta)$,
$\pi':=\pi(\theta+\tau)$, $\pi'_n:=\pi'_n(\theta+\tau)$, one obtains
\begin{align*}
 & \inner{\pi}{(J(\theta+\tau)-J(\theta))\pi'}=\\
 &\tau\left(
b_{n-1}\pi_{n-1}\pi'_n+\pi_n\left[b_{n-1}\pi'_{n-1} +
  \pi'_{n}(2\theta+\tau)(q_n+h)+b_n\pi'_{n+1}\right]
+\pi_{n+1}\pi'_nb_n\,.
\right)
\end{align*}
It follows from \eqref{eq:sym-relation-for-J} that
\begin{align*}
  &\frac{\lambda_k(\theta+\tau)-\lambda_k(\theta)}{\tau}=\\
&\frac{b_{n-1}\pi_{n-1}\pi'_n+\pi_n\left[b_{n-1}\pi'_{n-1} +
  \pi'_{n}(2\theta+\tau)(q_n+h)+b_n\pi'_{n+1}\right]
+\pi_{n+1}\pi'_nb_n}{\inner{\pi}{\pi'}}\,.
\end{align*}
Taking the limit $\tau\to 0$ and using
Lemma~\ref{lem:continuity-eigenvectors} and Remark~\ref{rem:entries-eigenvectors} we get the
proof of a).

Proof of b)\\
We proceed similarly. Note that
\begin{equation*}
    J(h+\delta)-J(h)=
\begin{pmatrix}
0 & 0 & 0 & 0 & 0 & 0 & \cdots \\[1mm]
0 & \ddots & \ddots & 0 & 0 & 0 & \cdots \\[1mm]
0  &  \ddots  & 0 & 0 & 0 & 0 & \cdots\\
0 & 0 & 0 & \theta^2\delta & 0 & 0 & \cdots\\
0 & 0 & 0 & 0 & 0 & 0 & \\
0 & 0 & 0 & 0 & 0 & 0 & \ddots\\
\vdots & \vdots & \vdots & \vdots & & \ddots & \ddots
  \end{pmatrix}\,.
\end{equation*}
Therefore,
\begin{equation*}
  \inner{\pi(h)}{(J(h+\delta)-J(h))\pi(h+\delta)}=\theta^2\delta\pi_n(h)\pi_n(h+\delta)\,,
\end{equation*}
where since $\theta$ remains fixed, we have not written the dependence
on $\theta$.
Using \eqref{eq:symmetric-relation}, one has
\begin{equation*}
  \frac{\lambda_k(h+\delta)-\lambda_k(h)}{\delta}=
\frac{\theta^2\pi_n(h)\pi_n(h+\delta)}{\inner{\pi(h)}{\pi(h+\delta)}}\,.
\end{equation*}
Taking the limit $\delta\to 0$
and using Lemma~\ref{lem:continuity-eigenvectors} and
Remark~\ref{rem:entries-eigenvectors}, we conclude the proof of b)
\end{proof}
The following result is well known (see \cite[Last Thm. in Sec.135]{MR1068530}).
\begin{proposition}
  \label{prop:riesz-nagy}
  Let $\{A_n\}_{n=1}^\infty$ and $A$ be self-adjoint operators all
  having the same domain and such that
  \begin{equation*}
    \norm{A_n-A}\convergesto{n} 0\,.
  \end{equation*}
Let $E_n$ and $E$ be the corresponding spectral families. If
$I$ is the interval $(\mu_1,\mu_2)$ with
$\mu_1,\mu_2\not\in\sigma(A)$, then
 \begin{equation*}
    \norm{E_n(I)-E(I)}\convergesto{n} 0\,.
  \end{equation*}
\end{proposition}
We shall need the following result
\begin{lemma}
  \label{lem:continuity-eigenvectors}
Assume  H\,1 and consider N\,1.
  \begin{enumerate}[a)]
  \item For any fixed $h\in\reals$,
    $\norm{\pi(\theta+\tau,h))-\pi(\theta,h)}\convergestozero{\tau} 0$.
 \item For any fixed $\theta>0$,
  $\norm{\pi(\theta,h+\delta))-\pi(\theta,h)}\convergestozero{\delta} 0$.
  \end{enumerate}
\end{lemma}
\begin{proof}
  We prove a), the proof of b) is analogous. From
  \eqref{eq:difference-theta-tau-theta} we see that
  \begin{equation*}
    \norm{J(\theta+\tau)-J(\theta)}\convergestozero{\tau} 0\,.
  \end{equation*}
Let $I$ be the interval  $(\mu_1,\mu_2)$ with
$\mu_1,\mu_2\not\in\sigma(J(\theta))$ and such that
$\sigma(J(\theta))\cap I=\{\lambda(\theta\}$. By the definition of
$\pi(\theta)$, we have
\begin{equation*}
  \pi(\theta)=kE_\theta(I)\delta_1\,,
\end{equation*}
where $k$ is a constant and $E_\theta$ is the spectral family of
$J(\theta)$. Since $1=\inner{kE_\theta(I)\delta_1}{\delta_1}$, one has
\begin{equation*}
  k=\frac{1}{\norm{E_\theta(I)}^2}\,.
\end{equation*}
Thus, if $\tau$ is small enough, then
\begin{equation*}
   \pi(\theta+\tau)=\frac{1}{\norm{E_{\theta+\tau}(I)}^2}E_{\theta+\tau}(I)\delta_1\,.
\end{equation*}
It follows from Proposition~\ref{prop:riesz-nagy} that
\begin{equation*}
  E_{\theta+\tau}(I)\convergestozeros{\tau}{\norm{\cdot}}E_\theta(I)\,,
\end{equation*}
therefore
\begin{equation*}
  \norm{\pi(\theta+\tau,h))-\pi(\theta,h)}\convergestozero{\tau} 0\,.
\end{equation*}
\end{proof}
\begin{remark}
  \label{rem:entries-eigenvectors}
Recalling that $\pi_n(\theta,h)$ is the polynomial of first kind
evaluated at $\lambda_k(\theta,h)$, one has, for fixed $h$,
\begin{equation*}
 \abs{\pi_n(\theta,h)-\pi_n(\theta+\tau,h)}\convergestozero{\tau} 0\,.
\end{equation*}
Indeed,
\begin{align*}
  \abs{\pi_n(\theta,h)-\pi_n(\theta+\tau,h)}&=\abs{\inner{\pi(\theta,h)}{\delta_n}-\inner{\pi(\theta+\tau,h)}{\delta_n}}\\
  &=\abs{\inner{\pi(\theta,h)-\pi(\theta+\tau,h)}{\delta_n}}\\
&\le\norm{\pi(\theta,h)-\pi(\theta+\tau,h)}\convergestozero{\tau}0\,.
\end{align*}
Analogously, for fixed $\theta$,
\begin{equation*}
 \abs{\pi_n(\theta,h)-\pi_n(\theta, h+\delta)}\convergestozero{\delta} 0\,.
\end{equation*}
\end{remark}
\begin{lemma}
  \label{lem:uniform-convergence-of-moments}
 Assume  H\,1 and consider N\,1.
  For fixed $h\in\reals$ and $n\in\nats$,
  \begin{equation}
    \label{eq:series-of-moments}
    \sum_{k\in M}\frac{\abs{\lambda_k(\theta,h)}^m}{\alpha_k(\theta,h)}
  \end{equation}
converges uniformly for $\theta\in [\theta_1,\theta_2]$. Moreover, for
fixed $\theta>0$, it converges uniformly for $h\in[h_1,h_2]$, where
\begin{equation*}
  \alpha_k(\theta,h):=\norm{\pi(\theta,h)}^2
\end{equation*}
is the normalizing constants corresponding to $\lambda_k(\theta,h)$.
\end{lemma}
\begin{proof}
  We consider the case when $h$
  is fixed. The case of $\theta$
  fixed is analogous. If $m$
  is even, the series \eqref{eq:series-of-moments} converges pointwise
  in $\theta$
  to the momentum $s_m:=\inner{\delta_1}{J^m(\theta)\delta_1}$
  which is a continuous function of $\theta$
  since it is the first entry of the matrix associated to
  $J(\theta)^m$.
  The terms of the series \eqref{eq:series-of-moments} are continuous
  in $\theta$
  as follows from Lemma~\ref{lem:derivative} and
  \ref{lem:continuity-eigenvectors}. Therefore, the series is
  uniformly convergent in an interval $[\theta_1,\theta_2]$
  (see \cite[Sec.\,1.31]{MR3155290}. If $m$ is odd
  \begin{equation*}
    \abs{\frac{\lambda_k(\theta,h)^m}{\alpha_k(\theta,h)}}\le
\frac{\lambda_k(\theta,h)^{m+1}}{\alpha_k(\theta,h)}
  \end{equation*}
  whenever $\abs{\lambda_k(\theta,h)}>1$.
  Hence, the series \eqref{eq:series-of-moments} also converges
  uniformly in this case for $\theta\in[\theta_1,\theta_2]$.
\end{proof}
\begin{theorem}
  \label{thm:absolute-convergence}
  Assume H\,1, consider N\,1, and abbreviate
  $\lambda_k:=\lambda_k(1,0)$ and $\mu_k:=\lambda_k(\theta,h)$. Then
  \begin{equation}
    \label{eq:absolute-sum}
    \sum_{n\in M}\abs{\mu_k-\lambda_k} <\infty\,.
  \end{equation}
\end{theorem}
\begin{proof}
  Let $A=[a,b]$, where $a:=\min\{\theta,1\}$,
  $b:=\max\{\theta,1\}$. By Lemma~\ref{lem:derivative} a), letting $h$
  fixed, we have
  \begin{equation*}
    \abs{\lambda_k(\theta)-\lambda_k(1)}\le
\int_A\frac{2\abs{\pi_n}}{\norm{\pi}^2}(\abs{b_{n-1}\pi_{n-1}}+\abs{b_{n}\pi_{n+1}} + \abs{\pi_{n}\theta(q_n+h)})d\theta\,.
  \end{equation*}
Taking a sequence $\{M_m\}$ as before, we get
\begin{equation*}
  \sum_{k\in M}\abs{\lambda_k(\theta)-\lambda_k(1)}\le
  2 \lim_{m\to\infty}\int_A\sum_{k\in M_m}\frac{\abs{\pi_n}}{\norm{\pi}^2}(\abs{b_{n-1}\pi_{n-1}}+\abs{b_{n}\pi_{n+1}} + \abs{\pi_{n}\theta(q_n+h)})d\theta\,.
\end{equation*}
Using Lemma~\ref{lem:uniform-convergence-of-moments}, and recalling
that $\pi_n$ are polynomials in $\lambda_k(\theta)$, we can get the
limit inside the integral and we obtain
\begin{equation*}
  \sum_{k\in M}\abs{\lambda_k(\theta)-\lambda_k(1)}<\infty
\end{equation*}
since the function inside the integral is continuous and $A$ is
compact. Note that we have fixed $h$.

Let now $B:=[c,d]$,
 where $c:=\min\{0,h\}$
and $d:\max\{0,h\}$.
Then, it follows from Lemma~\ref{lem:derivative} b) that
\begin{equation*}
  \abs{\lambda_k(1,h)-\lambda_k(1,0)}=\int_B\frac{\pi_n(1,h)^2}{\norm{\pi(1,h)}^2}dh\,.
\end{equation*}
In turn, one has
\begin{equation*}
  \sum_{k\in M}\abs{\lambda_k(1,h)-\lambda_k(1,0)}\le
  \lim_{m\to\infty}\int_A\sum_{k\in M_m}\frac{\abs{\pi_n(1,h)}}{\norm{\pi(1,h)}^2}dh\,.
\end{equation*}
We get the limit inside the integral since the convergence is uniform
by Lemma~\ref{lem:uniform-convergence-of-moments} and recalling that
$\pi_n$ is a polynomial evaluated at $\lambda_k(1,h)$. Therefore
\begin{equation*}
  \sum_{k\in M}\abs{\lambda_k(1,h)-\lambda_k(1,0)}\le
  \int_A\sum_{k\in M}\frac{\abs{\pi_n(1,h)}}{\norm{\pi(1,h)}^2}dh<\infty\,.
\end{equation*}
Now, apply the triangle inequality to obtain
\begin{equation*}
  \sum_{k\in M}\abs{\lambda_k(\theta,h)-\lambda_k(1,0)}\le\sum_{k\in
    M}\abs{\lambda_k(\theta,h)-\lambda_k(1,h)} +
\sum_{k\in M}\abs{\lambda_k(1,h)-\lambda_k(1,0)}\,.
\end{equation*}
\end{proof}
The previous result is related to \cite[Thm.\,II]{MR900507} which
states that the result holds for some enumeration.
\begin{proposition}
  \label{prop:zeros-poles-m}
  Assume H\,1.  The set of zeros of $\mathfrak{M}_k(z)$ coincides with
  $\sigma(\widetilde{J}_n)\setminus\sigma(J)$, while the set of poles
  is $\sigma(J)\setminus\sigma(\widetilde{J}_n)$.
\end{proposition}
\begin{proof}
  It follows from \cite[Cor.\,3.1]{MR3377115} that the zeros
  and poles of $\mathfrak{M}_k(z)$ are given by the poles of
  $\widetilde{G}(z,k)$ and $G(z,k)$ respectively. By
  Definition~\ref{def:green-function} any pole of $G(z,k)$ is an
  eigenvalue of $J$. If a pole of $G$ different from $\gamma$ were an
  eigenvalue of $\widetilde{J}_n$, then a contradiction follows from
  \cite[Lem.\,3.2]{MR3377115}. Note that, as a consequence
  of \cite[Thm.\,3.3]{MR3377115}, $\gamma$ is not a pole of
  $\mathfrak{M}_k(z)$. We have established that the poles of
  $\mathfrak{M}_k(z)$ are in
  $\sigma(J)\setminus\sigma(\widetilde{J}_n)$. The converse inclusion
  follows directly from \cite[Lem.\,2.2]{MR3377115}. By
  means of \cite[Lem.\,3.4 and Thm.\,3.3]{MR3377115} one
  proof that set of zeros of $\mathfrak{M}_k(z)$ equals
  $\sigma(\widetilde{J}_n)\setminus\sigma(J)$.
\end{proof}
% \begin{lemma}
%   \label{lem:mikhail}
%   \begin{equation}\label{coin}
% \sigma (J)\cap \sigma (\widetilde{J}_n)= \sigma (J)\cap (\{\lambda :G(\lambda ,n)=0\}\cup
% \{\gamma\})
% \end{equation}
% \begin{equation}\label{coincide}
% \card\left(\sigma (J)\cap \sigma (\widetilde{J}_n) \cap \{\lambda :G(\lambda,n)=0\} \right) \leq n.
% \end{equation}
% \begin{equation}\label{ora}
%  \lambda_r \in \sigma (J)\cap\{\lambda :G(\lambda,n)=0\}\Longrightarrow \mathfrak{M}_n(\lambda_r )
%  =\theta^2\end{equation}
% \begin{equation}\label{ora1}
% \lambda_r =\gamma\in\sigma (J)\Longrightarrow \mathfrak{M}_n (\lambda_r )
% =\theta^2+(1-\theta ^{2})|\psi_r(n) |^2\geq\theta^2
% \end{equation}
% \begin{equation}\label{ora2}
% \gamma\not\in\sigma (J)\Longrightarrow  \mathfrak{M}_n(\gamma)=\theta^2
% \end{equation}
% \begin{equation}\label{o!}
%  \gamma\in  \sigma(J) \cup \sigma(\widetilde{J}_n)\Longrightarrow \gamma \in \sigma(J)\cap  \sigma(\widetilde{J}_n).
% \end{equation}
% If $\lambda_N\ne \gamma$, then $\tilde\lambda_N\ne\lambda_N$
% and if $\lambda_1\ne \gamma$, then $\tilde\lambda_1\ne\lambda_1$.

% \end{lemma}

\begin{lemma}
  \label{lem:unif-convergence-prod}
  Let $\{\lambda_k\}_{k\in M}$ and $\{\mu_k\}_{k\in M}$ be sequences
  such that \eqref{eq:absolute-sum} holds. Let $\mathcal{B}=\cup_{k\in
  M}B_k$, where $B_k=\{z\in\complex: \abs{z-\lambda_k}<a\}$.
 Then
  \begin{equation*}
    \prod_{k\in M}\frac{z-\mu_k}{z-\lambda_k}
  \end{equation*}
 converges uniformly in $\complex\setminus\mathcal{B}$ and
 \begin{equation*}
   \lim_{\substack{\abs{z}\to\infty\\ z\not\in\mathcal{B}}}\prod_{k\in M}\frac{z-\mu_k}{z-\lambda_k}=1\,.
 \end{equation*}
\end{lemma}
\begin{proof}
  Since $z\not\in\mathcal{B}$,
  \begin{equation*}
    \abs{\frac{\lambda_k-\mu_k}{z-\lambda_k}}<\frac{1}{a}\abs{\lambda_k-\mu_k}
  \end{equation*}
for all $k\in M$. Therefore
\begin{equation*}
  \sum_{k\in M}\abs{\frac{\lambda_k-\mu_k}{z-\lambda_k}}
\end{equation*}
converges uniformly in $\complex\setminus\mathcal{B}$
\cite[Thm.\,7.10]{MR0385023}. In turn, this implies
\begin{equation*}
  \prod_{k\in M}\left(1+\frac{\lambda_k-\mu_k}{z-\lambda_k}\right)=
  \prod_{k\in M}\frac{z-\mu_k}{z-\lambda_k}
\end{equation*}
converges uniformly in in $\complex\setminus\mathcal{B}$ (see
for instance \cite[Sec.\,1.44]{MR3155290}).

Now, for $M_n$ such that $\card(M_n)=n$, one has
\begin{equation*}
  \abs{\prod_{k\in M}\frac{z-\mu_k}{z-\lambda_k}-1}\le
\abs{\prod_{k\in M}\frac{z-\mu_k}{z-\lambda_k}-\prod_{k\in M_n}\frac{z-\mu_k}{z-\lambda_k}}+
\abs{\prod_{k\in M_n}\frac{z-\mu_k}{z-\lambda_k}-1}\,.
\end{equation*}
If $n$ is sufficiently large the first term in the r.\,h.\,s. is
arbitrarily small uniformly in $z\not\in\mathcal{B}$. For large
$\abs{z}$, the second term is also arbitrarily small.
\end{proof}

\begin{proposition}
  \label{prop:m-throu-zeros-poles}
  Assume H\,1.  Fix an arbitrary $n\in\nats$.  Assume that
  $\sigma(J)\setminus\sigma(\widetilde{J}_n)=\{\lambda_k\}_{k\in M}$,
  where the sequence is strictly increasing. Then there is an
  enumeration of $\sigma(\widetilde{J}_n)\setminus\sigma(J)$ such that
  $\sigma(\widetilde{J}_n)\setminus\sigma(J)=\{\mu_k\}_{k\in M}$ and
 \begin{equation*}
   \mathfrak{M}_n(z)=
      \prod\limits_{k\in M}\frac{z-\mu_k}{z-\lambda_k}\,.
 \end{equation*}
\end{proposition}
\begin{proof}
 Assume that $J$ is not semibounded from above.
 Let  $\lambda_{k_0}$ be the eigenvalue
  nearest to zero.
  Since $G(z,n)$
  is a meromorphic Herglotz function, it follows from
  Proposition~\ref{prop:krein-levin} and Remark~\ref{rem:draw-terms} that
  \begin{equation}
 \label{eq:levin-herglotz-gen}
    G(z,n)=C \frac{(z-\eta_{k_0})(z-\eta_{k_0-1})}{(z-\lambda_{k_0})(z-\lambda_{k_0-1})}
  \prod_{\substack{k\in M\\k\ne k_0,k_0-1}} \left(1-\frac{z}{\eta_k}\right)
  \left(1-\frac{z}{\lambda_k}\right)^{-1}\,,
  \end{equation}
where $\{\eta_k\}_{k\in M}$ are the zeros of $G(z,n)$, $C>0$, and
\begin{equation*}
 % \label{eq:enum-zeros-poles}
  \lambda_k<\eta_k<\lambda_{k+1}\quad\forall k\in M\,.
\end{equation*}

Taking into account Proposition~\ref{prop:zero-green-zero-perturbed},
one also has from Proposition~\ref{prop:krein-levin} and
Remark~\ref{rem:draw-terms} that
  \begin{equation*}
    \widetilde{G}(z,n)=\widetilde{C}
    \frac{(z-\eta_{k_0-1})(z-\eta_{k_0})(z-\eta_{k_0+1})}
{(z-\mu_{k_0-1})(z-\mu_{k_0})(z-\mu_{k_0+1})}
  \prod_{\substack{k\in M\\k\ne k_0,k_0\pm 1}} \left(1-\frac{z}{\eta_k}\right)
  \left(1-\frac{z}{\mu_k}\right)^{-1}\,,
  \end{equation*}
where $\widetilde{C}>0$ and
\begin{equation*}
  %\label{eq:enum-zeros-poles}
  \mu_k<\eta_k<\mu_{k+1}\quad\forall k\in M.
\end{equation*}
Since the enumeration of the sequence
$\{\eta_k\}_{k\in M}$ does not change, we have taken into
account that $\mu_{k_0+1}$ could be zero.

For any values of the perturbative parameters $\theta$
and $h$,
the eigenvalue $\lambda_k(\theta,h)$
(see the notation introduced before Lemma~\ref{lem:derivative}) is
constrained between $\eta_{k-1}$
and $\eta_k$
which do not move as $\theta$
and $h$
change (see Proposition~\ref{prop:zero-green-zero-perturbed}).
Therefore the enumeration of the sequence $\{\mu_k\}_{k\in M}$
is such that $\lambda_k(\theta,h)=\mu_k$
for any values of the perturbative parameters.

Consider a sequence $\{M_n\}_{n=1}^\infty$ of subsets of $M$, such
that $M_n\subset M_{n+1}$ and $\cup_nM_n=M$. One has
  \begin{align}
    \label{eq:m-goth-krein}
    \mathfrak{M}_n(z)&=C'\prod_{j=k_0-1}^{k_0+1}\frac{z-\mu_j}{z-\lambda_j}
    \lim_{n\to\infty}\frac{\displaystyle
  \prod_{\substack{k\in M_n\\k\ne k_0,k_0\pm 1}} \left(1-\frac{z}{\eta_k}\right)
  \left(1-\frac{z}{\lambda_k}\right)^{-1}}{\displaystyle
  \prod_{\substack{k\in M_n\\k\ne k_0,k_0\pm 1}} \left(1-\frac{z}{\eta_k}\right)
  \left(1-\frac{z}{\mu_k}\right)^{-1}}\notag\\
%   &=C\frac{z-\mu_0}{z-\lambda_0}
%     \lim_{n\to\infty}\prod_{\substack{k\in M_n\\k\ne 0}}
% \left(1-\frac{z}{\mu_k}\right)
%   \left(1-\frac{z}{\lambda_k}\right)^{-1}\notag\\
  &=C'\prod_{j=k_0-1}^{k_0+1}\frac{z-\mu_j}{z-\lambda_j}
   \prod_{\substack{k\in M\\k\ne k_0,k_0\pm 1}}
\left(1-\frac{z}{\mu_k}\right)
  \left(1-\frac{z}{\lambda_k}\right)^{-1}\,.
  \end{align}
By Theorem~\ref{thm:absolute-convergence}, it follows that
\begin{equation}
  \prod_{\substack{k\in M\\k\ne k_0,k_0\pm 1}}\frac{\lambda_k}{\mu_k}\quad\text{converges.}
\end{equation}
Thus, one writes
  \begin{equation}
    \label{eq:infinite-product-equalities}
    \prod_{\substack{k\in M\\k\ne k_0,k_0\pm 1}}
\left(1-\frac{z}{\mu_k}\right)
  \left(1-\frac{z}{\lambda_k}\right)^{-1}=
\prod_{\substack{k\in M\\k\ne k_0,k_0\pm 1}}\frac{\lambda_k}{\mu_k}
\prod_{\substack{k\in M\\k\ne k_0,k_0\pm 1}}\frac{z-\mu_k}{z-\lambda_k}\,.
  \end{equation}
From (\ref{eq:G=-asymptotics}) it follows that
\begin{equation}
  \label{eq:function-tends-to-1}
\lim_{\substack{z\to\infty \\ 0<\abs{\arg z}<\pi}}
    \mathfrak{M}_k(z)=1\,.
  \end{equation}
On the other hand, according to Lemma~\ref{lem:unif-convergence-prod}
the second product of the r.\,h.\,s. tends to 1 along any curve away from the
spectrum. This implies, together with \eqref{eq:m-goth-krein}, that
% But, since the last product of the r.\,h.\,s. of
%   (\ref{eq:infinite-product-equalities}) converges uniformly and
%   (\ref{eq:kato}) holds, one has
% \begin{equation}
%   \label{eq:product-tends-to-1}
%     \lim_{\substack{z\to\infty \\
%         0<\abs{\arg z}<\pi}}
%     \prod_{k\in M}
%     \frac{z-\mu_k}{z-\lambda_k}=
%     \lim_{\substack{z\to\infty \\
%        0<\abs{\arg z}<\pi}}
%     \prod_{k\in M}
%     \left(1+\frac{\mu_k-\lambda_k}{\lambda_k-z}\right)=1\,.
% \end{equation}
% Thus, (\ref{eq:m-goth-krein}), (\ref{eq:infinite-product-equalities}),
% (\ref{eq:function-tends-to-1}), and (\ref{eq:product-tends-to-1})
% imply that
\begin{equation*}
  C'=\prod_{\substack{k\in M\\k\ne k_0,k_0\pm 1}}\frac{\mu_k}{\lambda_k}\,.
\end{equation*}
The proposition is then proven for the case when $J$ is not
semibounded from above. The case when $J$ is semibounded from above is
treated analogously.
\end{proof}
\begin{remark}
  \label{rem:found-theta}
  There is a simple expression for the quotient of the unperturbed and
  perturbed masses. Indeed,
  assume that $\gamma$ is not a pole of $G(z,n)$. If  $\lambda$ is a
  common eigenvalue or $\gamma$, then
  \eqref{eq:master} implies that
  $\mathfrak{M}_n(\lambda)=\theta^2$. Thus,
  Proposition~\ref{prop:m-throu-zeros-poles} below and
  \ref{eq:increment-mass-spring} yields
  \begin{equation*}
    \prod_{k\in
      M}\frac{\lambda-\mu_k}{\lambda-\lambda_k}=\frac{m_n}{m_n+\Delta m_n}
  \end{equation*}
cf. \cite[Rem.\,1]{MR2915295}. It is hard to imagine a more direct
relation between the eigenvalues and the perturbed mass.
\end{remark}

\section{Inverse spectral analysis}
\label{sec:inverse-spectral-analysis}
\begin{definition}
  \label{def:delta}
  For any Borel set $\mathcal{A}$, put
  \begin{equation*}
    \delta_\lambda(\mathcal{A})=\begin{cases}
      1 & \text{if } \lambda\in\mathcal{A}\\
      0 & \text{otherwise.}
    \end{cases}
  \end{equation*}
\end{definition}
\begin{lemma}
  \label{lem:berg}
  Let
  \begin{equation*}
    \mu=\sum_{\lambda\in M}a_\lambda\delta_\lambda\,,
  \end{equation*}
where $a_\lambda>0$ for any $\lambda\in M$  and $M\subset\reals$ is a discrete infinite set.
Assume that the polynomials are dense in $L_2(\reals,\mu)$. Let
\begin{equation*}
  \nu=\sum_{\lambda\in F}b_\lambda\delta_\lambda- \sum_{\lambda\in
    M_0\subset M}a_\lambda\delta_\lambda\,,
\end{equation*}
where $b_\lambda>0$ for any $\lambda\in F$, and $F$ and $M_0$ are
finite sets of the same cardinality. Then the polynomials are dense in $L_2(\reals,\mu+\nu)$.
\end{lemma}
\begin{proof}
  Since the polynomials are dense in $L_2(\reals,\mu)$, $\mu$ is
  either indeterminate $N$-extremal or determinate by
  \cite[Thm.\,2.3.3 Cor.\,2.3.3]{MR0184042}. By \cite[Remark, p. 231]{MR1343638} the measure
  \begin{equation*}
    \mu-a_{\lambda_0}\delta_{\lambda_0}\,,\qquad\lambda_0\in M_0\,,
  \end{equation*}
is determinate (see also in \cite{MR1308001} Lemma B and the comment
before Lemma D). Now, by \cite[Thm. 5 (d) and (e)]{MR1343638}, the measure
\begin{equation*}
  \mu - a_{\lambda_0}\delta_{\lambda_0} +
  b_{\lambda}\delta_{\lambda}\,,\quad\lambda\in F\,.
\end{equation*}
is either determinate or indeterminate $N$-extremal.
\end{proof}
\begin{remark}
  \label{rem:thanks}
  If the cardinality of $M_0$ is less than the cardinality of $F$,
  then the previous lemma may not be true. In fact, adding just one
  point mass to the measure $\mu$ may destroy the density of the
  polynomials. For related results see \cite{MR3407921}.
 We thank C. Berg and A. Duran for making this fact
  clear to us and M. Sodin for pertinent remarks.
\end{remark}
\begin{proposition}
  \label{prop:G-reconstruction}
  Assume H\,1.  If $\gamma\not\in\sigma(J)$, then the sequences
  $\sigma(J)\setminus\sigma(\widetilde{J}_n)$ and
  $\sigma(\widetilde{J}_n)\setminus\sigma(J)$ together with the
  parameter $\gamma$ uniquely determine the function $G(z,n)$. If
  $\gamma\in\sigma(J)$, then the sequences
  $\sigma(J)\setminus\sigma(\widetilde{J}_n)$ and
  $\sigma(\widetilde{J}_n)\setminus\sigma(J)$ together with the
  parameters $\theta$ and $\gamma$ uniquely determine the function
  $G(z,n)$.
\end{proposition}
\begin{proof}
  By Proposition~\ref{prop:m-throu-zeros-poles}, one constructs the
  function $\mathfrak{M}_n(z)$. Thus, since
  $\gamma\not\in\sigma(J)$, $\gamma$ is not a pole of $G(z,n)$ and,
  then $\mathfrak{M}_n(\gamma)=\theta^2$. With the knowledge of
  $\gamma$ and $\theta$, one finds  $G(z,n)$ from
  (\ref{eq:master}). For the second assertion
   we use again Proposition \ref{prop:m-throu-zeros-poles} to construct
    $\mathfrak{M}_n(z)$ and then find $G(z,n)$ using (\ref{eq:master}).
\end{proof}
\begin{remark}
  \label{rem:g-not-sufficient}
  As it will be clear later (see
  Theorem~\ref{thm:necessary-sufficient}), although the noncommon
  eigenvalues are sufficient for reconstructing the $n$-Green function
  of the Jacobi operator, they are not sufficient for reconstructing
  the operator itself.
\end{remark}
\begin{theorem}
  \label{thm:reconstruction-from-two-spectra}
  Assume  H\,1.
  Let the spectra of $J$ and $\widetilde{J}_n$ and the
  parameter $\gamma$ be given so that $\gamma$ is not in the
  spectrum of $J$.
  \begin{enumerate}[(1)]
  \item If  the spectra do not
    intersect, then the data given determine $\theta$ uniquely  and
    there are countably many
    pairs Jacobi operators $J$, $\widetilde{J}_n$ with the given spectra.
  \item If the spectra intersect, then the data given determine
    $\theta$ uniquely and there are uncountable many pairs $J$,
    $\widetilde{J}_n$ with the given spectra.
%   \item $\gamma$ belongs to $\sigma(J)$ and
%     $\sigma(J)\cap\sigma(\widetilde{J}_n)\setminus\{\gamma\}\ne\emptyset$.
%   \begin{enumerate}
%     \item $\mathfrak{M}_n(\gamma)=\mathfrak{M}_n(\lambda)$ with
%       $\lambda\in\sigma(J)\cap\sigma(\widetilde{J}_n)$. Then one has $\binom{\aleph_0}{n-q-1}$ disjoint manifolds of dimension $q+1$ of $J$'s.
%    \item $\mathfrak{M}_n(\gamma)>\mathfrak{M}_n(\lambda)$ with
%      $\lambda\in\sigma(J)\cap\sigma(\widetilde{J}_n)$. Then one has $\binom{\aleph_0}{n-q}$ disjoint manifolds of dimension $q$ of $J$'s.
%   \end{enumerate}
% \item  $\gamma$ belongs to $\sigma(J)$  and
%     $\sigma(J)\cap\sigma(\widetilde{J}_n)\setminus\{\gamma\}=\emptyset$.
% % \begin{enumerate}
% % \item $\mathfrak{M}_n'(\gamma)\ne0$. Then the data given together with
% %   $\theta$ yield  $\binom{\aleph_0}{n}$  of $J$'s.
% % \item $\mathfrak{M}'(\gamma)=0$
%   \begin{enumerate}
%   \item if $\mathfrak{M}_n(\gamma)=\theta^2$, then the data given
%     yield $\binom{\aleph_0}{n-1}$ disjoint manifolds of dimension $1$ of $J$'s.
%   \item if $\mathfrak{M}_n(\gamma)>\theta^2$, then the data given together with $\theta$
%     yield  $\binom{\aleph_0}{n}$  of $J$'s.
%   \end{enumerate}
%     \end{enumerate}
\end{enumerate}
\end{theorem}
\begin{proof}
  (1) As in the proof of Proposition~\ref{prop:G-reconstruction} one
  recovers $\theta$ and, then, the function $G(z,n)$.  By
  Propositions \ref{prop:Gkk-formula} and
  \ref{prop:form-inverse-green-function}, one has
  \begin{equation}
    \label{eq:inverse-green-reconstruction}
    z-q_n+b_n^2m_n^+(z)+b_{n-1}^2m_n^-(z)=z-q_n+\sum_{k\in M}\frac{\eta_k}{\alpha_k-z}\,.
  \end{equation}
  The fact that the spectra do not intersect means that $m_n^-$ and
  $m_n^+$ do not posses common poles due to
  Proposition~\ref{prop:common-eigenvalues-+-}. Therefore,
  $\{\alpha_k\}_{k\in M}$ is the union of the disjoint sets
  $\{\text{poles of }m_n^+\}$ and $\{\text{poles of }m_n^-\}$. Thus,
  any choice of $n-1$ terms in the series of the r.\,h.\,s of
  \eqref{eq:inverse-green-reconstruction} can be made to correspond to
  $m_n^-$ and the infinite sum containing the remaining terms
  corresponds to $m_n^+$  due to
  Lemma~\ref{lem:berg}. Indeed, Lemma~\ref{lem:berg} shows that if
\begin{equation*}
\sum_{k\in M\setminus N}\frac{\eta_k}{\alpha_k-z}
\end{equation*}
is an Weyl $m$-function, then
\begin{equation*}
\sum_{k\in M\setminus N'}\frac{\eta_k}{\alpha_k-z}
\end{equation*}
is an Weyl $m$-function too, for any other set $N'$ such that $\card(N)=\card(N')$.
Hence
  \begin{equation*}
    b_{n-1}^2m_n^-=\sum_{k\in N}\frac{\eta_k}{\alpha_k-z}\,,\qquad
    b_n^2m_n^+=\sum_{k\in M\setminus N}\frac{\eta_k}{\alpha_k-z}\,.
  \end{equation*}
In view of the fact that
the null moment of the spectral measure of a Jacobi operator is $1$,
that is,
\begin{equation*}
  1=\sum_{k\in M\setminus N}b_n^{-2}\eta_k=\sum_{k\in N}b_{n-1}^{-2}\eta_k
\end{equation*}
(cf. Theorem~\ref{thm:nec-suf-cond-for-meausure}), one has
\begin{equation*}
   b_{n-1}^2:=\sum_{k\in N}\eta_k\quad\text{ and }
\quad b_n^2:=\sum_{k\in M\setminus N}\eta_k\,.
\end{equation*}
Having found $b_n$ and $b_{n-1}$, one finds $m_n^+$ and $m_n^-$. By
Remark~\ref{rem:inverseJM}, these functions determine $J_n^{+}$ and $J_n^{-}$.
The Jacobi operator $\mathfrak{J}$ defined by $J_n^{+}$, $J_n^{-}$, $b_{n-1}$,
$b_{n}$, $q_n$ has the function $G(z,n)$ as its Green function. We have found as many
$\mathfrak{J}$'s with this Green function as subsets of $n-1$ elements from the
countably set $M$. To complete the proof, it remains to show that the
spectra of $\mathfrak{J}$ and $\widetilde{\mathfrak{J}}_n$ coincide with the given
sequences. From the Green function $G(z,n)$ of $J$, construct
$\widehat{\mathfrak{M}}_n$ by \eqref{eq:master}. This function
coincides with $\mathfrak{M}_n$ obtained from the spectra of $J$ and
$\widetilde{J}_n$ since the Green function is the same with $\gamma$
given and $\theta$ uniquely determined. Therefore,
$\widehat{\mathfrak{M}}_n$ and $\mathfrak{M}_n$ have the same poles
and zeros which are $\sigma(J)$ and $\sigma(\widetilde{J}_n)$, respectively.

(2) As in the previous item, the data given allow to find $\theta$
and $G(z,n)$, so one has \eqref{eq:inverse-green-reconstruction}. Let
$C\subset M$ such that $k\in C$ whenever
$\alpha_k\in\sigma(J)\cap\sigma(\widetilde{J}_n)$. By Remark~\ref
{rem:finite-common-eigenvalues}, $\card(C)\le n-1$. Pick an arbitrary
$S\subset M\setminus C$ of $n-1-\card(C)$ elements and for each
$k\in C$ choose $\beta_k\in(0,1)$.  Thus, it follows from
\eqref{eq:inverse-green-reconstruction} that
  \begin{equation*}
    b_{n-1}^2m_n^-=\sum_{k\in S}\frac{\eta_k}{\alpha_k-z}+
\sum_{k\in C}\frac{\beta_k\eta_k}{\alpha_k-z}
\end{equation*}
and
  \begin{equation*}
    b_n^2m_n^+=\sum_{k\in (M\setminus C)\setminus S}\frac{\eta_k}{\alpha_k-z}+
\sum_{k\in C}\frac{(1-\beta_k)\eta_k}{\alpha_k-z}\,.
\end{equation*}
Using again the fact that the null moment of the spectral
measure of a Jacobi operator is $1$, one has
\begin{equation*}
   b_{n-1}^2:=\sum_{k\in S}\eta_k+\sum_{k\in C}\beta_k\eta_k\quad\text{ and }
\quad b_n^2:=\sum_{k\in (M\setminus C)\setminus S}\eta_k+
\sum_{k\in C}(1-\beta_k)\eta_k\,.
\end{equation*}
The Jacobi operator $\mathfrak{J}$ given by $J_n^{+}$, $J_n^{-}$,
$b_{n-1}$, $b_{n}$, $q_n$ has the Green function $G(z,n)$. Note that,
for any choice of the set $S$, there are as many solutions as elements
in the interval $(0,1)$. To conclude the proof, observe that by
construction the common eigenvalues of $J_n^{+}$ and $J_n^{-}$
coincide with $\sigma(J)\cap\sigma(\widetilde{J}_n)$. As in (1), we
show that $\widehat{\mathfrak{M}}_n(z)=\mathfrak{M}_n(z)$ which
implies that the noncommon eigenvalues of $J$ and $\widetilde{J}$
coincide the ones of $\mathfrak{J}$ and $\widetilde{\mathfrak{J}}_n$.

% In this case the functions $m_n^\pm$ have common poles
%  which are also the common eigenvalues of $J$ and
% $\widetilde{J}_n$ (see Proposition~\ref{prop:common-eigenvalues-+-}
% and \cite[Lem.\,3.5]{MR3377115}). Thus, one has to choose only
% $\binom{\aleph_0}{n-q}$ poles of $\widehat{m}_n^-$ among the poles of
% $-G^{-1}$ since the other poles are given. Here, for the common poles
% of the functions $\widehat{m}_n^\pm$, one has to share the
% residues of $-G^{-1}$ at these poles among the residues of
% $\widehat{m}_n^\pm$. Take any $\widehat{\beta}_k\in(0,1)$ for all $k$
% such that $\alpha_k$ is a common pole and let $\widehat{\beta}_k=1$
% for all other $k$'s. Put
% \begin{equation*}
%    \widehat{b}_{n-1}^2:=\sum_{k\in\widehat{\mathcal{N}}}\widehat{\beta}_k\eta_k\quad\text{ and }
% \quad\widehat{b}_n^2:=\sum_{k\in\nats\setminus\widehat{\mathcal{N}}}(1-\widehat{\beta}_k)\eta_k
% \end{equation*}
% It is now evident that $\widehat{J}$ is a solution of the inverse
% problem but this time one has $\binom{\aleph_0}{n-q}$ disjoint
% manifolds of dimension $q$ of such solutions.
\end{proof}
\begin{theorem}
  \label{thm:reconstruction-from-two-spectra1}
  Assume H\,1.  Let the spectra of $J$ and $\widetilde{J}_n$ and the
  parameters $\theta$ and $h$ be given. If
  $(\sigma(J)\cap\sigma(\widetilde{J}_n)\setminus\{\gamma\}=\emptyset$
  and $\mathfrak{M}_n(\gamma)>\theta^2$, then there are countably many
  solutions of the inverse problem. In all other cases when $\gamma$
  is in the spectrum of $J$, there are uncountably many solutions of
  the inverse problem.
  \end{theorem}
  \begin{proof}
    If we have $\theta$ and $h$ then by (\ref{eq:def-gamma}) we get
    $\gamma$ and using Proposition \ref{prop:G-reconstruction} we
    recover $G(z,n)$. To prove there are countably many solutions we
    shall show that there are no common poles of $m_n^-$ and $m_n^+$
    and the analysis will proceed then completely analogous to $(1)$
    of Theorem \ref{thm:reconstruction-from-two-spectra}.  Since the
    spectra of $J$ and $\widetilde{J}_n$ in this case only intersect
    at $\gamma$, this is the only possible common pole. This follows
    from Proposition \ref{prop:common-eigenvalues-+-}, recalling the
    spectra of $J_n^+$ and $J_n^-$ are the poles of $m_n^-$ and
    $m_n^+$ respectively. But the condition
    $\mathfrak{M}_n(\gamma)>\theta^2$ together with with
    \eqref{eq:master} imply that $\gamma$ is a pole of $G(z,n)$ and
    therefore $\gamma$ cannot be a pole of $m_n^-$ or $m_n^+$ since
    $G(z,n)$ vanishes at these poles (cf. Proposition
    \ref{prop:Gkk-formula}). We conclude that there are no common
    poles of $m_n^-$ and $m_n^+$ as was to be shown. In all other
    cases, when $\gamma\in\sigma(J)\cap\sigma(\widetilde{J}_n)$, the
    Weyl $m$-functions $m_n^-$ and $m_n^+$ have common eigenvalues and the
    analysis is analogous to $(2)$ of Theorem
    \ref{thm:reconstruction-from-two-spectra}. Notice that if
    $\gamma\in\sigma(J)$ and not a pole of $G(z,n)$ , then
    $G(\gamma,n)=0$ and $\gamma$ is a common pole of $m_n^-$ and
    $m_n^+$ by \cite[Cor.\,2.3,
    Lems.\,2.8,\,2.9]{MR3377115}.
  \end{proof}
\begin{remark}
  \label{rem:theorem-mikhail}
  The previous theorem corresponds to \cite[Thm.\,4]{MR2915295} which
  deals with the
  case of finite mass-spring systems. In that setting, the analogous
  of case
  (1) of
  Theorem~\ref{thm:reconstruction-from-two-spectra} yields a finite
  set of solutions. The other cases in
  \cite[Thm.\,4]{MR2915295} can also be treated in a way similar to
  (1) and (2) of
  Theorem~\ref{thm:reconstruction-from-two-spectra}.
\end{remark}
\begin{theorem}
  \label{thm:necessary-sufficient}
  Let $S$ and $\widetilde{S}$ be two infinite sequences without finite
  points of accumulation, $\gamma\in\reals$ with
  $\gamma\not\in S\cup\widetilde{S}$, and $n\in\nats$. There is a
  matrix \eqref{eq:jm-0} such that $S=\sigma(J)$ and
  $\widetilde{S}=\sigma(\widetilde{J}_n)$, with $0<\theta<1$ and
  $h=\gamma\left(1/\theta^2-1\right)$, if and only if the following
  conditions hold
  \begin{enumerate}[(1)]
  \item Between two consecutive points of
    $(S\setminus\widetilde{S})\cup\{\gamma\}$ there is exactly one point
    of $\widetilde{S}\setminus S$. Any point of
    $\widetilde{S}\setminus S$ lies between two consecutive points of
    $(S\setminus\widetilde{S})\cup\{\gamma\}$.

  If the strictly
    increasing sequence $\{\lambda_k\}_{k\in
      M}$ coincides with $S\setminus\widetilde{S}$, then we enumerate the interlacing
    points $\{\mu_k\}_{k\in
      M}=\widetilde{S}\setminus S$ such that for $\beta\in
    (S\setminus\widetilde{S})\cup\{\gamma\}$,
    \begin{align*}
        \lambda_k<\mu_k<\beta &\quad \text{ if }\, \beta\le\gamma
      \quad \text{ and }\quad
(\lambda_k,\beta)\cap (S\setminus\widetilde{S})=\emptyset\,,\\
      \beta<\mu_k<\lambda_k & \quad \text{ if }\, \beta\ge \gamma
      \quad \text{ and }\quad (\beta,\lambda_k)\cap
                              (S\setminus\widetilde{S})
=\emptyset\,.
    \end{align*}
  \item The series
    \begin{equation*}
       \sum_{k\in M} \abs{\lambda_k-\mu_k}
    \end{equation*}
is convergent.
\item If $\lambda\in S\cap\widetilde{S}$, then
  $\mathfrak{N}(\lambda)=\mathfrak{N}(\gamma)$, where
\begin{equation}
  \label{eq:N-definition}
  \mathfrak{N}(z):=\prod_{k\in M}\frac{z-\mu_k}{z-\lambda_k}\,.
\end{equation}
\item The function $\mathfrak{G}$ satisfies condition 3 of
  Proposition~\ref{prop:nec-suf-green}, where
  \begin{equation*}
    \mathfrak{G}(z):=
\frac{\mathfrak{N}(z)-\mathfrak{N}(\gamma)}{(\mathfrak{N}(\gamma)-1)(z-\gamma)}\,.
  \end{equation*}
\end{enumerate}
\end{theorem}
\begin{proof}
  We begin by proving that the conditions are necessary, i.\,e., given
  a matrix \eqref{eq:jm-0} such that $S=\sigma(J)$ and
  $\widetilde{S}=\sigma(\widetilde{J}_n)$ with $\theta<1$ and
  $h=\gamma\left(1/\theta^2-1\right)$, then the sequences satisfy
  Conditions (1)--(4). Condition (1) is a consequence of
  \cite[Thms.\,3.8 and 3.10]{MR3377115}. Condition (2) follows from
  Theorem~\ref{thm:absolute-convergence}. For this, we have to show
  that the set
  $\sigma(\widetilde{J}_n)\setminus\sigma(J)=\{\mu_k\}_{k\in M}$ is
  such that $\mu_k=\lambda_k(\theta,h)$ (as in
  Theorem~\ref{thm:absolute-convergence}). Indeed, for any values of
  the perturbative parameters $\theta$ and $h$, the eigenvalue
  $\lambda_k(\theta,h)$ (see the notation introduced before
  Lemma~\ref{lem:derivative}) is constrained between $\eta_{k-1}$ and
  $\eta_k$ which do not move as $\theta$ and $h$ change (see
  Proposition~\ref{prop:zero-green-zero-perturbed}).  Therefore the
  enumeration of the sequence $\{\mu_k\}_{k\in M}$ is such that
  $\lambda_k(\theta,h)=\mu_k$ for any values of the perturbative
  parameters.

% For each $\lambda_k\in\sigma(J)\setminus\sigma(\widetilde{J}_n)$
% consider the function $\lambda_k(\theta,h)$ such that
% $\lambda_k(\theta,h)\in\sigma(J(\theta,h))$ and
% $\lambda_k(1,0)=\lambda_k$. The function $\lambda_k(\theta,h)$ is
% a continuous function of $\theta$ and $h$
% (Lemma~\ref{lem:derivative}). If $\lambda_k(1,0)=\lambda_{k'}

  For proving Condition (3), observe that \eqref{eq:master} and
  Proposition~\ref{prop:m-throu-zeros-poles} give
  \begin{equation*}
   \mathfrak{N}(z)=
\prod_{k\in M}\frac{z-\mu_k}{z-\lambda_k}=\theta^2+(1-\theta^2)(\gamma-z)G(z,n)\,.
  \end{equation*}
Now, \cite[Lem.\,3.5]{MR3377115} implies that
$\mathfrak{N}(\lambda)=\theta^2$ for any $\lambda\in
S\cap\widetilde{S}$. Also, from \eqref{eq:master}, it follows that
$\mathfrak{N}(\gamma)=\theta^2$ since $\gamma\not\in S$ means that
$\gamma$ is not a pole of $G(z,n)$.   Condition (4) follows from
  Proposition~\ref{prop:m-throu-zeros-poles}, \eqref{eq:master} and
  Proposition~\ref{prop:nec-suf-green}.

Let us now prove that the conditions are sufficient.
First observe that Condition (1) implies that
\begin{equation*}
  0<\frac{\gamma-\mu_k}{\gamma-\lambda_k}<1
\end{equation*}
for any $k\in M$. Thus,
\begin{equation*}
  0<\mathfrak{N}(\gamma)<1\,.
\end{equation*}

From \eqref{eq:N-definition}, one has
  \begin{equation}
\label{eq:prospect-herglotz}
    -\frac{\mathfrak{N}(z)}{z-\gamma}
=-\lim_{n\to\infty}(z-\gamma)^{-1}\prod_{k\in M_n}\frac{z-\mu_k}{z-\lambda_k}\,.
  \end{equation}
% where we have chosen the sequence $\{M_n\}_{n=1}^\infty$ in such a way
% that $k_0\in M_n$ for any $n\in\nats$ (see  Proposition~\ref{prop:interlacing}).
Expanding in partial fractions
  \begin{equation}
    \label{eq:finite-herglotz}
    -\prod_{k\in
      M_n}\frac{z-\mu_k}{z-\nu_k}\frac{1}{z-\gamma}=\sum_{k\in
      M_n}\frac{\alpha_k}{\lambda_k-z}+
\frac{\alpha}{\gamma-z}\,.
  \end{equation}
  Since the sequences $\{\mu_k\}_{k\in M}$ and
  $\{\lambda_k\}_{k\in M}\cup\{\gamma\}$ interlace according to
  Condition (1), we have that $\alpha_k>0$ for any $k\in M_n$ and
  $\alpha>0$ and, therefore, the l.\,h.\,s. of
  \eqref{eq:finite-herglotz} is Herglotz (cf. the proof of
  \cite[Cor.\,2.5]{MR1616422}) for any $n\in\nats$. By definition, the
  l.\,h.\,s. of \eqref{eq:prospect-herglotz} is analytic outside the
  real axis and it is the limit of Herglotz functions, thus it is a
  Herglotz function. By \cite[Thm.\,2, Chap.\,7]{MR589888}, one has
\begin{align}\label{rpo:cl-represent-chevotarev}
 -\frac{\mathfrak{N}(z)}{z-\gamma}= a z + b + \sum_{\substack{k \in \mathcal{M}\\ k
  \not=0}} A_{k}\left(\frac{1}{\nu_k - z}-
\frac{1}{\nu_k}\right) + \frac{A_0}{\nu_0-z},
\end{align} where $a \geq 0$, $b \in \mathbb{R}$ and $A_k \geq 0$ for
all $k \in \mathcal{M}$. Here only $\nu_0$ is allowed
to be zero. Assume, without loss of generality, that
$\gamma=\nu_{k_0}$ and $k_0\ne 0$. Then
\begin{equation*}
  -\frac{\mathfrak{N}(z)}{z-\gamma}= a z + b + \sum_{\substack{k \in \mathcal{M}\\ k
  \ne0,k_0}} A_{k}\left(\frac{1}{\nu_k - z}-
\frac{1}{\nu_k}\right) + \frac{A_0}{\nu_0-z} +  A_{k_0}\left(\frac{1}{\gamma - z}-
\frac{1}{\gamma}\right)\,.
\end{equation*}
Note that
\begin{equation}
 \label{eq:residue-N}
  \res_{z=\gamma}\frac{\mathfrak{N}(z)}{z-\gamma}:=\lim_{z\to\gamma}(z-\gamma)
\frac{\mathfrak{N}(z)}{z-\gamma}=\mathfrak{N}(\gamma)\,.
\end{equation}
Since the residue of the r.\,h.\,s. at $\gamma=\nu_{k_0}$ of
\eqref{rpo:cl-represent-chevotarev} is $-A_{k_0}$ (see the proof of \cite[Chap.\,7,
Thm.\,2]{MR589888}), one has
\begin{equation*}
  -\frac{\mathfrak{N}(z)}{z-\gamma}+\frac{\mathfrak{N}(\gamma)}{z-\gamma}=
a z + b + \sum_{\substack{k \in \mathcal{M}\\ k
  \ne0,k_0}} A_{k}\left(\frac{1}{\nu_k - z}-
\frac{1}{\nu_k}\right) + \frac{A_0}{\nu_0-z} -  A_{k_0}\frac{1}{\gamma}\,.
\end{equation*}
The last equation implies that 
\begin{equation*}
-\frac{\mathfrak{N}(z)}{z-\gamma}+\frac{\mathfrak{N}(\gamma)}{z-\gamma}
\end{equation*}
is a Herglotz function. This, in turn, yields that $\mathfrak{G}$ is
Herglotz since $0<\mathfrak{N}(\gamma)<1$.

Note that, as was shown in Lemma~\ref{lem:unif-convergence-prod},
Condition (2) implies that
\begin{equation*}
   1=\lim_{\substack{\abs{z}\to\infty\\
       z\not\in\mathcal{B}}}\prod_{k\in M}\frac{z-\mu_k}{z-\lambda_k}
=\lim_{\substack{\abs{z}\to\infty\\
       z\not\in\mathcal{B}}}\mathfrak{N}(z)
 \end{equation*}
Which in turn yields
\begin{equation*}
  \mathfrak{G}(z)=-\frac{1}{z}+O(z^{-2})
\end{equation*}
as $z\to\infty$
along any curve away of a strip containing the real axis.  Taking into
account Condition (4), it follows from
Proposition~\ref{prop:nec-suf-green}, that $\mathfrak{G}$
is the $n$-th Green function of a family of Jacobi operators.
There is an element $J$ of this family such that $\sigma(J)=S$. Indeed,
by \eqref{eq:Gkk-formula2}, there are $q_n\in\reals$, $b_n,b_{n-1}>0$,
and Weyl $m$-functions $m_n^\pm$, such that
\begin{equation*}
  -\mathfrak{G}(z)^{-1}=z-q_n+b_n^2m_n^++b_{n-1}^2m_n^-\,.
\end{equation*}
Now for the spectrum of $J$ to be $S$, one can always choose the Weyl
$m$-functions of the submatrices $J_n^\pm$ in such a way that
$m_n^\pm$ have common poles at $S\cap\widetilde{S}$.  Here Condition
(3) guarantees that $\mathfrak{G}$ has zeros at $S\cap\widetilde{S}$.

Having fixed the operator $J$ such that $\sigma(J)=S$, one defines
\begin{equation}
  \label{eq:theta-definition}
  \theta:=+\sqrt{\mathfrak{N}(\gamma)}
\quad\text{ and }\quad h:=\gamma\left(1/\theta^2-1\right)\,.
\end{equation}
Note that $\theta<1$. Consider the Jacobi operator $J$ and
the operator $\widetilde{J}_n$ given in \eqref{eq:def-tilde-j}. The
$n$-th Green function of $J$ satisfies
\begin{equation*}
  G(z,n)=\frac{\prod_{k\in
      M}\frac{z-\widetilde{\mu}_k}{z-\lambda_k}-\theta^2}
{(\theta^2-1)(z-\gamma)}\,,
\end{equation*}
where $\{\widetilde{\mu}_k\}_{k\in M}=\sigma(\widetilde{J}_n)$ (see \eqref{eq:master}. From
the definition of $\mathfrak{G}$ and our definition of $\theta$
(see \eqref{eq:theta-definition}), one has
\begin{equation*}
  \prod_{k\in
      M}\frac{z-\widetilde{\mu}_k}{z-\lambda_k}=\prod_{k\in
      M}\frac{z-\mu_k}{z-\lambda_k}
\end{equation*}
for any $z\in\complex\setminus\reals$. The last equality implies that
\begin{equation*}
  \mu_k=\widetilde{\mu}_k\quad\text{for any}\quad k\in M\,.
\end{equation*}

\end{proof}
\begin{theorem}
  \label{thm:necessary-sufficient-intersection}
  Let $S$ and $\widetilde{S}$ be two infinite sequences without finite
  points of accumulation, $\gamma\in S\cap\widetilde{S}$, and
  $n\in\nats$. There is a matrix \eqref{eq:jm-0} such that
  $S=\sigma(J)$ and $\widetilde{S}=\sigma(\widetilde{J}_n)$, with
  $0<\theta<1$ and $h=\gamma\left(1/\theta^2-1\right)$, if and only if
  the following conditions hold
  \begin{enumerate}[(1)]
  \item Between two consecutive points of
    $(S\setminus\widetilde{S})\cup\{\gamma\}$ there is exactly one point
    of $\widetilde{S}\setminus S$. Any point of
    $\widetilde{S}\setminus S$ lies between two consecutive points of
    $(S\setminus\widetilde{S})\cup\{\gamma\}$.

  If the strictly
    increasing sequence $\{\lambda_k\}_{k\in
      M}$ coincides with $S\setminus\widetilde{S}$, then we enumerate the interlacing
    points $\{\mu_k\}_{k\in
      M}=\widetilde{S}\setminus S$ such that for $\beta\in
    (S\setminus\widetilde{S})\cup\{\gamma\}$,
    \begin{align*}
        \lambda_k<\mu_k<\beta &\quad \text{ if }\, \beta\le\gamma
      \quad \text{ and }\quad
(\lambda_k,\beta)\cap (S\setminus\widetilde{S})=\emptyset\,,\\
      \beta<\mu_k<\lambda_k & \quad \text{ if }\, \beta\ge \gamma
      \quad \text{ and }\quad (\beta,\lambda_k)\cap
                              (S\setminus\widetilde{S})
=\emptyset\,.
    \end{align*}
  \item The series
    \begin{equation*}
       \sum_{k\in M} \abs{\lambda_k-\mu_k}
    \end{equation*}
is convergent.
\item If $\lambda_1,\lambda_2$ are in
  $(S\cap\widetilde{S})\setminus\{\gamma\}$, then
  $\mathfrak{N}(\lambda_1)=\mathfrak{N}(\lambda_2)\le\mathfrak{N}(\gamma)$, where
  \begin{equation*}
    \mathfrak{N}(z):=\prod_{k\in M}\frac{z-\mu_k}{z-\lambda_k}\,.
  \end{equation*}
\item The function
 \begin{equation*}
    \mathfrak{G}(z):=
    \begin{cases}
      \frac{\mathfrak{N}(z)-\mathfrak{N}(\omega)}{(\mathfrak{N}(\omega)-1)(z-\gamma)}
      & \text{ if } (S\cap\widetilde{S})\setminus\{\gamma\}\ne\emptyset\\
\frac{\mathfrak{N}(z)-\vartheta^2}{(\vartheta^2-1)(z-\gamma)}& \text{ if } (S\cap\widetilde{S})\setminus\{\gamma\}=\emptyset\,,
    \end{cases}
  \end{equation*}
where
\begin{equation*}
\omega\in (S\cap\widetilde{S})\setminus\{\gamma\}\,,
\quad \vartheta^2\in
\begin{cases}
  (0,\mathfrak{N}(\gamma)) & \text{ if } \mathfrak{N}'(\gamma)\ne 0\\
(0,\mathfrak{N}(\gamma)] & \text{ if } \mathfrak{N}'(\gamma)= 0
\end{cases}
\end{equation*}
satisfies condition 3 of
  Proposition~\ref{prop:nec-suf-green}.
\end{enumerate}
\end{theorem}
\begin{proof}
  The fact that (1) and (2) are necessary is proven as in
  Theorem~\ref{thm:necessary-sufficient}. Let us prove that (3) is
  necessary. Note that if $S=\sigma(J)$ and
  $\widetilde{S}=\sigma(\widetilde{J}_n)$, then
  $\mathfrak{N}(z)=\mathfrak{M}_n(z)$. By \cite[Lem.\,3.5]{MR3377115},
  $G(\lambda_1,n)=G(\lambda_2,n)=0$ for any
  $\lambda_1,\lambda_2\in(S\cap\widetilde{S})\setminus\{\gamma\}$. Thus,
  \eqref{eq:master} implies that
  $\mathfrak{M}_n(\lambda_1)=\mathfrak{M}_n(\lambda_2)=\theta^2$. Now,
  it follows from \cite[Lem.\,2.8]{MR3377115}, that $\gamma$ is either
  a zero or a pole of $G(z,n)$, therefore by \eqref{eq:green-integral}
  \begin{equation*}
    \res_{z=\gamma}G(z,n)=-\pi^2_n(\gamma)\rho\{\gamma\}\le 0\,.
  \end{equation*}
 Thus, using \eqref{eq:master}, one has
 \begin{equation*}
   \mathfrak{M}_n(\gamma)=\theta^2-(1-\theta^2) \res_{z=\gamma}G(z,n)\ge\theta^2\,.
 \end{equation*}
 Condition (4) follows from Proposition~\ref{prop:nec-suf-green} and
 the fact that, due to Condition~(3), the function $\mathfrak{G}(z)$ coincides
 with $G(z,n)$.

 We now prove that the conditions are sufficient. First note that
 Condition~(1) implies that
 \begin{equation}
0<\mathfrak{N}(\gamma)<1.\label{eq:less-then-one}
\end{equation}
As in the
 proof of Theorem~\ref{thm:necessary-sufficient}, one shows that
 \begin{equation*}
      -\frac{\mathfrak{N}(z)}{z-\gamma}
 \end{equation*}
is a Herglotz function. Thus, assuming  without loss of generality, that
$\gamma=\nu_{k_0}$ and $k_0\ne 0$. Then
\begin{equation*}
  -\frac{\mathfrak{N}(z)}{z-\gamma}= a z + b + \sum_{\substack{k \in \mathcal{M}\\ k
  \ne0,k_0}} A_{k}\left(\frac{1}{\nu_k - z}-
\frac{1}{\nu_k}\right) + \frac{A_0}{\nu_0-z} + \mathfrak{N}(\gamma)\left(\frac{1}{\gamma - z}-
\frac{1}{\gamma}\right)\,,
\end{equation*}
where we have used \ref{eq:residue-N}. Hence
\begin{equation}
\label{eq:first-candidate-herglotz}
\begin{split}
  &-\frac{\mathfrak{N}(z)}{z-\gamma} + \frac{\vartheta^2}{z-\gamma} = \\
&=a z + b + \sum_{\substack{k \in \mathcal{M}\\ k
  \ne0,k_0}} A_{k}\left(\frac{1}{\nu_k - z}-
\frac{1}{\nu_k}\right) + \frac{A_0}{\nu_0-z} + (\mathfrak{N}(\gamma)-\vartheta^2)\left(\frac{1}{\gamma - z}-
\frac{1}{\gamma}\right)\,.
\end{split}
\end{equation}
In the case $(S\cap\widetilde{S})\setminus\{\gamma\}=\emptyset$, since
$\theta^2\le\mathfrak{N}(\gamma)$, as required in Condition (4),
the l.\,h.\,s. of \eqref{eq:first-candidate-herglotz} is
a Herglotz function. When
$(S\cap\widetilde{S})\setminus\{\gamma\}\ne\emptyset$, one analogously
obtains that
\begin{equation*}
  -\frac{\mathfrak{N}(z)}{z-\gamma} + \frac{\mathfrak{N}(\omega)}{z-\gamma}
\end{equation*}
is a Herglotz function due to Condition (3). By
\eqref{eq:less-then-one}, taking into account Conditions (3) and (4),
one concludes that $\mathfrak{G}$ is a Herglotz function. The rest of
the proof is the same as the part of the proof of
Theorem~\ref{thm:necessary-sufficient} after it is established that
$\mathfrak{G}$ is a Herglotz function.
\end{proof}
\begin{remark}
  \label{rem:other-theta}
  When the peturbation parameter $\theta$ is greater than 1, one can
  prove results along the same lines as
  Theorems~\ref{thm:necessary-sufficient} and
  \ref{thm:necessary-sufficient-intersection}. In this case the point
  $\gamma$ acts as a ``repeller'' instead of being an ``attractor''.
\end{remark}

\textbf{Open problems.} We have just scratched the surface of
some inverse spectral theorems for Jacobi operators and many questions
remain open. In the model studied here, we would like to know how many
perturbations are needed to recover the system uniquely. If the perturbation
takes place in the first mass, then just the spectral information
provided by two spectra is enough. How many spectra do we need if the
perturbation happens at the $n$-th mass? How do we determine from the
spectral information where the perturbation took place? How about
reconstruction results when we have partial information of the spectra
or when Hypothesis 1 above does not hold?
\subsection*{Acknowledgments}
We thank C. Berg, A. Dur\'an, F. Marcellan and M. Sodin for valuable remarks, and
R. del Rio A. for a hint to the literature.

\appendix
For reader's convenience we give the proof of the following
assertion which follows from a result due to M. G. Krein
\cite[Chap.\,7,\,Thm.\,1]{MR589888} (cf. \cite[Sec.\,4]{MR2305710}).
\begin{proposition}
  \label{prop:krein-levin}
   Let $m(z)$ be the Weyl $m$-function of a Jacobi operator with
   discrete spectrum. Then
   \begin{enumerate}[i)]
   \item the zeros and poles of $m(z)$ are real, simple and interlace,
   \item the zeros $\{\eta_k\}_{k\in M}$ and poles
     $\{\lambda_k\}_{k\in M}$ of $m(z)$ can be enumerated in such a
     way that if $0\in\{\eta_k\}_{k\in M}\cup\{\lambda_k\}_{k\in M}$,
     then either $\eta_{k_0}=0$ or $\lambda_{k_0}=0$ and
  \begin{equation}
 \label{eq:levin-herglotz-gen}
    m(z)=C \frac{z-\eta_{k_0}}{z-\lambda_{k_0}}
  \prod_{\substack{k\in M\\k\ne k_0}} \left(1-\frac{z}{\eta_k}\right)
  \left(1-\frac{z}{\lambda_k}\right)^{-1}\,,
  \end{equation}
where $C<0$ and
\begin{equation}
  \label{eq:enum-zeros-poles-alt}
  \eta_k<\lambda_k<\eta_{k+1}\quad\forall k\in M
\end{equation}
if $\sigma(J)$ is semi-bounded from above,
while, $C>0$ and
\begin{equation}
  \label{eq:enum-zeros-poles}
  \lambda_k<\eta_k<\lambda_{k+1}\quad\forall k\in M
\end{equation}
otherwise.
   \end{enumerate}
\end{proposition}
\begin{proof}
  Item i) follows from the fact that $m(z)$ is a (nonconstant) Herglotz meromorphic
  function and the argument principle (see the proof of
  \cite[Chap.\,7,\,Thm.\,1]{MR589888}). For proving ii) first we show
  that the infinite product in (\ref{eq:levin-herglotz-gen}) converges
  uniformly on compacts not containing the poles. Note that, due to
  the interlacing property, one has
  \begin{equation*}
    0<\sum_{k\in M\setminus\{k_0\}}\left(\frac{1}{\eta_k}-\frac{1}{\lambda_k}\right)<\sum_{k\in M\setminus\{k_0\}}\left(\frac{1}{\eta_k}-\frac{1}{\eta_{k+1}}\right)\,
  \end{equation*}
for the \eqref{eq:enum-zeros-poles-alt} case and
  \begin{equation*}
    0<\sum_{k\in M\setminus\{k_0\}}\left(\frac{1}{\lambda_k}-\frac{1}{\eta_k}\right)<\sum_{k\in M\setminus\{k_0\}}\left(\frac{1}{\lambda_k}-\frac{1}{\lambda_{k+1}}\right)\,
  \end{equation*}
  for the \ref{eq:enum-zeros-poles} case.  These inequalities imply
  that the series
  \begin{equation*}
    \sum_{k\in M\setminus\{k_0\}}\left(\frac{1}{\lambda_k}-\frac{1}{\eta_k}\right)
  \end{equation*}
converges in both \eqref{eq:enum-zeros-poles-alt} and \ref{eq:enum-zeros-poles} cases.
 Now, since
\begin{equation*}
  \sum_{k\in M\setminus\{k_0\}}\left[\left(1-\frac{z}{\eta_k}\right)
  \left(1-\frac{z}{\lambda_k}\right)^{-1}-1\right]<z\sum_{k\in M\setminus\{k_0\}}\left(\frac{1}{\lambda_k}-\frac{1}{\eta_k}\right)\left(1-\frac{z}{\lambda_k}\right)^{-1}\,
\end{equation*}
the infinite product in (\ref{eq:levin-herglotz-gen}) converges
  uniformly on compacts not containing $\lambda_k$ for any $k\in
  M\setminus\{k_0\}$.
As in the proof of \cite[Chap.\,7,\,Thm.\,1]{MR589888}, one can show
that when \ref{eq:enum-zeros-poles} holds, the function
\begin{equation*}
  \frac{z-\eta_{k_0}}{z-\lambda_{k_0}}
  \prod_{\substack{k\in M\\k\ne k_0}} \left(1-\frac{z}{\eta_k}\right)
  \left(1-\frac{z}{\lambda_k}\right)^{-1}
\end{equation*}
is Herglotz due to
\begin{equation*}
  0<\arg \left(\frac{1 - \frac{z}{\eta_k}}{{1-
        \frac{z}{\lambda_k}}}\right)=\measuredangle\lambda_k z\eta_k<\pi
\end{equation*}
and
\begin{equation*}
  0<\sum_{k\in M}\measuredangle\lambda_k z\eta_k\le\pi\,.
\end{equation*}
Analogously, if \eqref{eq:enum-zeros-poles-alt} takes place,
\begin{equation*}
  -\frac{z-\eta_{k_0}}{z-\lambda_{k_0}}
  \prod_{\substack{k\in M\\k\ne k_0}} \left(1-\frac{z}{\eta_k}\right)
  \left(1-\frac{z}{\lambda_k}\right)^{-1}
\end{equation*}
is Herglotz since now
\begin{equation*}
  0>\measuredangle\lambda_k z\eta_k>-\pi\,.
\end{equation*}
 Following the same reasoning as in the proof of
 \cite[Chap.\,7,\,Thm.\,1]{MR589888}, one shows that
 \begin{equation*}
   \frac{m(z)}{\pm\frac{z-\eta_{k_0}}{z-\lambda_{k_0}}
  \prod_{\substack{k\in M\\k\ne k_0}} \left(1-\frac{z}{\eta_k}\right)
  \left(1-\frac{z}{\lambda_k}\right)^{-1}}
 \end{equation*}
 is a constant. For finishing the proof it remains to show that
 \ref{eq:enum-zeros-poles} occurs when $J$ is not semibounded from
 above and \eqref{eq:enum-zeros-poles-alt} happens when $J$ is
 semibounded from above. If $J$ is semibounded from below, $J_1^+$
 also is. Moreover $J\le J_1^+$, so the
 $\min\sigma(J)\le\min\sigma(J_1^+)$. Since the poles of $m(z)$
 constitute the spectrum of $J$ and the set of zeros is
 $\sigma(J_1^+)$ one has that the smallest pole is less than the
 smallest zero, i.\,e. \ref{eq:enum-zeros-poles}. If $J$ is
 semibounded from above, one has $J_1^+<J$ and, then, the biggest pole
 is greater than the biggest zero, that is
 \eqref{eq:enum-zeros-poles-alt}.
\end{proof}
\begin{remark}
  \label{rem:draw-terms}
  From the proof of the previous theorem, it follows that there is a
  positive constant $C$ such that
  \begin{equation*}
    m(z)=C\prod_{j=k_0-n}^{k=k_0+n}\frac{z-\eta_{j}}{z-\lambda_{j}} \prod_{k\in M\setminus\{k_0-n,\dots,k_0+n\}} \left(1-\frac{z}{\eta_k}\right)
  \left(1-\frac{z}{\lambda_k}\right)^{-1}
  \end{equation*}
for any $n\in\nats$. This is used in the proof of Proposition~\ref{prop:m-throu-zeros-poles}.
\end{remark}
\def\cprime{$'$} \def\lfhook#1{\setbox0=\hbox{#1}{\ooalign{\hidewidth
  \lower1.5ex\hbox{'}\hidewidth\crcr\unhbox0}}} \def\cprime{$'$}
  \def\cprime{$'$} \def\cprime{$'$} \def\cprime{$'$} \def\cprime{$'$}
  \def\cprime{$'$} \def\cprime{$'$}


\begin{thebibliography}{10}

\bibitem{MR0184042}
N.~I. Akhiezer.
\newblock {\em The classical moment problem and some related questions in
  analysis}.
\newblock Translated by N. Kemmer. Hafner Publishing Co., New York, 1965.

\bibitem{MR1255973}
N.~I. Akhiezer and I.~M. Glazman.
\newblock {\em Theory of linear operators in {H}ilbert space}.
\newblock Dover Publications Inc., New York, 1993.
\newblock Translated from the Russian and with a preface by Merlynd Nestell,
  Reprint of the 1961 and 1963 translations, Two volumes bound as one.

\bibitem{MR0222718}
J.~M. Berezans{\cprime}ki{\u\i}.
\newblock {\em Expansions in eigenfunctions of selfadjoint operators}.
\newblock Translated from the Russian by R. Bolstein, J. M. Danskin, J. Rovnyak
  and L. Shulman. Translations of Mathematical Monographs, Vol. 17. American
  Mathematical Society, Providence, R.I., 1968.

\bibitem{MR1308001}
C.~Berg and A.~J. Duran.
\newblock The index of determinacy for measures and the {$l^2$}-norm of
  orthonormal polynomials.
\newblock {\em Trans. Amer. Math. Soc.}, 347(8):2795--2811, 1995.

\bibitem{MR1192782}
M.~S. Birman and M.~Z. Solomjak.
\newblock {\em Spectral theory of selfadjoint operators in {H}ilbert space}.
\newblock Mathematics and its Applications (Soviet Series). D. Reidel
  Publishing Co., Dordrecht, 1987.
\newblock Translated from the 1980 Russian original by S. Khrushch{\"e}v and V.
  Peller.

\bibitem{MR2263317}
M.~T. Chu and G.~H. Golub.
\newblock {\em Inverse eigenvalue problems: theory, algorithms, and
  applications}.
\newblock Numerical Mathematics and Scientific Computation. Oxford University
  Press, New York, 2005.

\bibitem{MR504044}
C.~de~Boor and G.~H. Golub.
\newblock The numerically stable reconstruction of a {J}acobi matrix from
  spectral data.
\newblock {\em Linear Algebra Appl.}, 21(3):245--260, 1978.

\bibitem{MR2915295}
R.~del Rio and M.~Kudryavtsev.
\newblock Inverse problems for {J}acobi operators: {I}. {I}nterior mass-spring
  perturbations in finite systems.
\newblock {\em Inverse Problems}, 28(5):055007, 18, 2012.

\bibitem{MR2998707}
R.~del Rio, M.~Kudryavtsev, and L.~O. Silva.
\newblock Inverse problems for {J}acobi operators {III}: {M}ass-spring
  perturbations of semi-infinite systems.
\newblock {\em Inverse Probl. Imaging}, 6(4):599--621, 2012.

\bibitem{MR3113459}
R.~del Rio, M.~Kudryavtsev, and L.~O. Silva.
\newblock Inverse problems for {J}acobi operators {II}: {M}ass perturbations of
  semi-infinite mass-spring systems.
\newblock {\em Zh. Mat. Fiz. Anal. Geom.}, 9(2):165--190, 277, 281, 2013.

\bibitem{MR3377115}
R.~del Rio and L.~O. Silva.
\newblock Spectral analysis for linear semi-infinite mass-spring systems.
\newblock {\em Math. Nachr.}, 288(11-12):1241--1253, 2015.

\bibitem{MR1045318}
M.~G. Gasymov and G.~S. Guse{\u\i}nov.
\newblock On inverse problems of spectral analysis for infinite {J}acobi
  matrices in the limit-circle case.
\newblock {\em Dokl. Akad. Nauk SSSR}, 309(6):1293--1296, 1989.

\bibitem{MR1616422}
F.~Gesztesy and B.~Simon.
\newblock {$m$}-functions and inverse spectral analysis for finite and
  semi-infinite {J}acobi matrices.
\newblock {\em J. Anal. Math.}, 73:267--297, 1997.

\bibitem{MR2102477}
G.~M.~L. Gladwell.
\newblock {\em Inverse problems in vibration}, volume 119 of {\em Solid
  Mechanics and its Applications}.
\newblock Kluwer Academic Publishers, Dordrecht, second edition, 2004.

\bibitem{MR0447294}
L.~J. Gray and D.~G. Wilson.
\newblock Construction of a {J}acobi matrix from spectral data.
\newblock {\em Linear Algebra and Appl.}, 14(2):131--134, 1976.

\bibitem{MR499269}
G.~{\v{S}}. Guse{\u\i}nov.
\newblock The determination of the infinite {J}acobi matrix from two spectra.
\newblock {\em Mat. Zametki}, 23(5):709--720, 1978.

\bibitem{MR0221315}
R.~Z. Halilova.
\newblock An inverse problem.
\newblock {\em Izv. Akad. Nauk Azerba\u\i d\v zan. SSR Ser. Fiz.-Tehn. Mat.
  Nauk}, 1967(3-4):169--175, 1967.

\bibitem{MR0213379}
H.~Hochstadt.
\newblock On some inverse problems in matrix theory.
\newblock {\em Arch. Math. (Basel)}, 18:201--207, 1967.

\bibitem{MR0382314}
H.~Hochstadt.
\newblock On the construction of a {J}acobi matrix from spectral data.
\newblock {\em Linear Algebra and Appl.}, 8:435--446, 1974.

\bibitem{MR549425}
H.~Hochstadt.
\newblock On the construction of a {J}acobi matrix from mixed given data.
\newblock {\em Linear Algebra Appl.}, 28:113--115, 1979.

\bibitem{MR900507}
T.~Kato.
\newblock Variation of discrete spectra.
\newblock {\em Comm. Math. Phys.}, 111(3):501--504, 1987.

\bibitem{MR3407921}
M.~Langer and H.~Woracek.
\newblock Stability of {N}-extremal measures.
\newblock {\em Methods Funct. Anal. Topology}, 21(1):69--75, 2015.

\bibitem{MR589888}
B.~J. Levin.
\newblock {\em Distribution of zeros of entire functions}, volume~5 of {\em
  Translations of Mathematical Monographs}.
\newblock American Mathematical Society, Providence, R.I., revised edition,
  1980.
\newblock Translated from the Russian by R. P. Boas, J. M. Danskin, F. M.
  Goodspeed, J. Korevaar, A. L. Shields and H. P. Thielman.

\bibitem{mono-marchenko}
V.~A. Marchenko.
\newblock {\em Introduction to the theory of inverse problems of spectral
  analysis}.
\newblock Universitetski Lekcii. Akta, Kharkov, 2005.
\newblock In Russian.

\bibitem{MR1463594}
P.~Nylen and F.~Uhlig.
\newblock Inverse eigenvalue problem: existence of special spring-mass systems.
\newblock {\em Inverse Problems}, 13(4):1071--1081, 1997.

\bibitem{MR1436689}
P.~Nylen and F.~Uhlig.
\newblock Inverse eigenvalue problems associated with spring-mass systems.
\newblock In {\em Proceedings of the {F}ifth {C}onference of the
  {I}nternational {L}inear {A}lgebra {S}ociety ({A}tlanta, {GA}, 1995)}, volume
  254, pages 409--425, 1997.

\bibitem{MR1343638}
M.~P{\'e}rez~Riera and J.~L. Varona~Malumbres.
\newblock On completeness of orthogonal systems and {D}irac deltas.
\newblock {\em J. Comput. Appl. Math.}, 58(2):225--231, 1995.

\bibitem{MR1247178}
Y.~M. Ram.
\newblock Inverse eigenvalue problem for a modified vibrating system.
\newblock {\em SIAM J. Appl. Math.}, 53(6):1762--1775, 1993.

\bibitem{MR1068530}
F.~Riesz and B.~Sz.-Nagy.
\newblock {\em Functional analysis}.
\newblock Dover Books on Advanced Mathematics. Dover Publications, Inc., New
  York, 1990.
\newblock Translated from the second French edition by Leo F. Boron, Reprint of
  the 1955 original.

\bibitem{MR1307384}
M.~Rosenblum and J.~Rovnyak.
\newblock {\em Topics in {H}ardy classes and univalent functions}.
\newblock Birkh\"auser Advanced Texts: Basler Lehrb\"ucher. [Birkh\"auser
  Advanced Texts: Basel Textbooks]. Birkh\"auser Verlag, Basel, 1994.

\bibitem{MR0385023}
W.~Rudin.
\newblock {\em Principles of mathematical analysis}.
\newblock McGraw-Hill Book Co., New York-Auckland-D\"usseldorf, third edition,
  1976.
\newblock International Series in Pure and Applied Mathematics.

\bibitem{MR2305710}
L.~O. Silva and R.~Weder.
\newblock On the two spectra inverse problem for semi-infinite {J}acobi
  matrices.
\newblock {\em Math. Phys. Anal. Geom.}, 9(3):263--290 (2007), 2006.

\bibitem{MR2438732}
L.~O. Silva and R.~Weder.
\newblock The two-spectra inverse problem for semi-infinite {J}acobi matrices
  in the limit-circle case.
\newblock {\em Math. Phys. Anal. Geom.}, 11(2):131--154, 2008.

\bibitem{smith1998infrared}
B.~Smith.
\newblock {\em Infrared Spectral Interpretation: A Systematic Approach}.
\newblock Taylor \& Francis, 1998.

\bibitem{spletzer-et-al1}
M.~Spletzer, A.~Raman, H.~Sumali, and J.~P. Sullivan.
\newblock Highly sensitive mass detection and identification using vibration
  localization in coupled microcantilever arrays.
\newblock {\em Appl. Phys. Lett.}, 92:114102, 2008.

\bibitem{spletzer-et-al2}
M.~Spletzer, A.~Raman, A.~Q. Wu, and X.~Xu.
\newblock Ultrasensitive mass sensing using mode localization in coupled
  microcantilevers.
\newblock {\em Appl. Phys. Lett.}, 88:254102, 2006.

\bibitem{MR1643529}
G.~Teschl.
\newblock Trace formulas and inverse spectral theory for {J}acobi operators.
\newblock {\em Comm. Math. Phys.}, 196(1):175--202, 1998.

\bibitem{MR1711536}
G.~Teschl.
\newblock {\em Jacobi operators and completely integrable nonlinear lattices},
  volume~72 of {\em Mathematical Surveys and Monographs}.
\newblock American Mathematical Society, Providence, RI, 2000.

\bibitem{MR3155290}
E.~C. Titchmarsh.
\newblock {\em The theory of functions}.
\newblock Oxford University Press, Oxford, 1958.
\newblock Reprint of the second (1939) edition.

\end{thebibliography}
\end{document}